\documentclass[12pt,english]{article}
\usepackage[T1]{fontenc}
\usepackage[latin9]{inputenc}
\usepackage{geometry}
\usepackage{color}
\usepackage{mathrsfs}
\usepackage{amsmath}
\usepackage{amsthm}
\usepackage{graphicx}
\usepackage{amssymb}
\usepackage{setspace}
\usepackage{multirow}
\usepackage[authoryear]{natbib}
\usepackage{xargs}[2008/03/08]
\usepackage[explicit]{titlesec}
\titleformat{\section}[hang]{\color{black}\filcenter}{}{0pt}{\large \bf \sf \thesection.\quad \uppercase{#1}}%
\titleformat{name=\section,numberless}[hang]{\color{black}\filcenter}{}{0pt}{\large \bf \sf \quad \uppercase{#1}}%
\titleformat{\subsection}[hang]{\bfseries}{}{0pt}{\large \sf \thesubsection. \quad{#1}}
\titleformat{name=\subsection,numberless}[hang]{\bfseries}{}{0pt}{\large \sf {#1}}

\usepackage{afterpage}

\newcommand{\single}{\renewcommand{\baselinestretch}{1.0}\normalsize}
\newcommand{\double}{\renewcommand{\baselinestretch}{1.63}\normalsize}
\newcommand{\bc}{\begin{center}}
\newcommand{\ec}{\end{center}}
\newcommand{\no}{\noindent}

\makeatletter
\theoremstyle{plain}
\newtheorem{thm}{\protect\theoremname}
  \theoremstyle{remark}
  
  \theoremstyle{plain}
  \newtheorem{lem}[thm]{\protect\lemmaname}

\@ifundefined{date}{}{\date{}}
\usepackage{ifpdf} 
\ifpdf 

 \IfFileExists{lmodern.sty}{\usepackage{lmodern}}{}

\fi 

\usepackage{mathrsfs}
\usepackage{enumitem}
\usepackage{ifthen}
\usepackage{xcolor}
\usepackage{natbib} 
\def\LyX{\texorpdfstring{%
  L\kern-.1667em\lower.25em\hbox{Y}\kern-.125emX\@}
  {LyX}}

\newcommand{\mylabel}[2]{#2\def\@currentlabel{#2}\label{#1}}
\newcommand{\vertiii}[1]{{\left\vert\kern-0.25ex\left\vert\kern-0.25ex\left\vert #1 
    \right\vert\kern-0.25ex\right\vert\kern-0.25ex\right\vert}}

\makeatletter
\title{Statistical Analysis of Longitudinal Data on Riemannian Manifolds}
\let\Title\@title

\usepackage[hidelinks]{hyperref}
\hypersetup{
    colorlinks,
    linkcolor={red!50!black},
    citecolor={blue!50!black},
    urlcolor={blue!80!black}
}

\makeatother

\usepackage{babel}
  \providecommand{\lemmaname}{Lemma}
  \providecommand{\remarkname}{Remark}
\providecommand{\theoremname}{Theorem}
\providecommand{\corname}{Corollary}

\def\references{\bibliography{manifold,fda}}
\def\ci{\cite}
\def\cp{\citep}
\newcommand{\be}{\begin{eqnarray}}
\newcommand{\ee}{\end{eqnarray}}
\newcommand{\bsp}{\begin{split}}
\newcommand{\esp}{\end{split}}
\newcommand{\ed}{\end{document}}

\newcommand{\btab}{\begin{tabular}}
\newcommand{\etab}{\end{tabular}}

\newcommand{\la}{\label}

\def\bco{\iffalse}

\global\long\def\expect{\mathbb{E}}
\global\long\def\prob{\mathrm{Pr}}

\global\long\def\real{\mathbb{R}}
\global\long\def\Op{O_{P}}
\global\long\def\manifold{\mathcal{M}}

\global\long\def\op{o_{P}}

\global\long\def\asympeq{\asymp}

\global\long\def\diffop{\mathrm{d}}
\newcommandx\tangentspace[2][usedefault, addprefix=\global, 1=\manifold]{T_{#2}#1}

\global\long\def\dim{d}

\newcommandx\lpnorm[3][usedefault, addprefix=\global, 1=r, 2=]{\|#3\|_{\mathcal{L}^{#1}}^{#2}}
\newcommandx\lp[1][usedefault, addprefix=\global, 1=p]{\mathcal{L}^{#1}}
\global\long\def\transpose{\mathrm{T}}

\global\long\def\innerprod#1#2{\langle#1,#2\rangle}

\global\long\def\Log{\mathrm{Log}}
\global\long\def\Exp{\mathrm{Exp}}
\newcommandx\vfnorm[3][usedefault, addprefix=\global, 1=\mu, 2=]{\|#3\|_{#1}^{#2}}
\newcommandx\vfinnerprod[2][usedefault, addprefix=\global, 1=\mu]{\llangle#2\rrangle_{#1}}
\global\long\def\define{:=}

\global\long\def\tdomain{\mathcal{T}}

\newcommandx\opnorm[3][usedefault, addprefix=\global, 1=\mu, 2=]{\vertiii{#3}_{#1}^{#2}}
\newcommandx\fronorm[2][usedefault, addprefix=\global, 1=]{|#2|_{F}^{#1}}

\global\long\def\mmax{m_{\max{}}}
\global\long\def\hxi{\hat{\xi}}
\global\long\def\txi{\tilde{\xi}}
\global\long\def\bLi{\mathbf{L}_i}
\global\long\def\hX{\hat{X}}
\global\long\def\hL{\hat{L}}
\global\long\def\hmu{\hat{\mu}}
\global\long\def\bLi{\mathbf{L}_i}
\global\long\def\hbLi{\hat{\mathbf{L}}_i}
\global\long\def\btLi{\tilde{\mathbf{L}}_i}
\global\long\def\hphi{\hat{\phi}}
\global\long\def\bphi{\boldsymbol{\phi}}
\global\long\def\hbphi{\hat{\bphi}}
\global\long\def\hlambda{\hat{\lambda}}
\global\long\def\hsigma{\hat{\sigma}}
\global\long\def\bSigLi{\boldsymbol{\Sigma}_{\bLi}}
\global\long\def\hbSigLi{\hat{\boldsymbol{\Sigma}}_{\bLi}}
\global\long\def\blup{\mathcal{B}}
\global\long\def\Vec{\mathrm{Vec}}

\global\long\def\Tr{\mathrm{Tr}}
\DeclareMathOperator{\expm}{expm}

\begin{document}


\thispagestyle{empty} 

\vspace{-0.1in}

\single \bc {\bf \sc \Large Modeling Longitudinal Data on Riemannian Manifolds}
\vspace{0.2in}\\
Xiongtao Dai$^1$, Zhenhua Lin$^2$ and Hans-Georg M\"uller$^2$ \\

\vspace{.04in}

Department of Statistics, Iowa State University, Ames, IA 50011 USA\\
Department of Statistics, University of California, Davis\\
Davis, CA 95616 USA 
\ec \centerline{October 2018}

\vspace{0.2in} \thispagestyle{empty}
\bc{\bf \sf ABSTRACT} \ec \vspace{-.1in} \no 
\setstretch{1.1}
When considering  functional principal component analysis for sparsely observed
longitudinal  data that take values on a nonlinear manifold, a  major challenge  is how to handle the sparse and irregular observations that
are commonly encountered in longitudinal studies. 
Addressing this challenge,  
we provide theory and implementations for 
a manifold version of the principal analysis by conditional expectation (PACE)  procedure that produces representations intrinsic to the manifold, extending a well-established version 
of functional principal component analysis targeting sparsely sampled longitudinal data in linear spaces. 
Key steps are local linear smoothing methods for the
estimation of a Fr\'{e}chet mean curve, mapping the  observed manifold-valued longitudinal data  to tangent
spaces around the estimated mean curve, and applying smoothing methods 
to obtain the covariance structure of the mapped data.
Dimension reduction is achieved via  
representations based on the first few leading principal components.
A finitely
truncated representation of the original manifold-valued data is then obtained by mapping these
tangent space representations to the manifold. 
We show that the proposed estimates of mean curve
and covariance structure achieve state-of-the-art convergence
rates. 
For longitudinal emotional well-being data for unemployed workers as an example of time-dynamic compositional data that are located on a sphere,  we demonstrate that our methods lead to interpretable eigenfunctions and principal component scores, which are defined on tangent spaces.  In a second example, we analyze the body shapes of wallabies by mapping the relative size of their body parts onto a spherical pre-shape space. Compared to standard functional principal component analysis, which is based on Euclidean geometry,   the proposed approach leads to improved  trajectory recovery for sparsely sampled data on nonlinear manifolds.

\vspace{.15in}
\no {KEYWORDS:\quad Longitudinal Compositional Data, Data on Spheres, Dimension Reduction, Functional Data Analysis, Principal Component Analysis, Sparse and Irregular Data.}
\vspace{.1in}
\thispagestyle{empty} \vfill
\no {\small Research supported by NSF Grant DMS-1712864} 

\newpage
\pagenumbering{arabic} \setcounter{page}{1}

\double

\section{Introduction}\label{sec:Introduction}

Functional data are usually considered as elements of a Hilbert space \cp{horv:12,Hsing2015,mull:16:3}, a linear space with Euclidean geometry, where typical tools include 
functional principal component analysis \cp{Kleffe1973,Hall2006,chen:15:1}
and functional regression \cp{Hall2007c,kong:16,knei:16}.
Considerably  less work has been done on the analysis of  nonlinear functional data, which  are increasingly encountered 
in practice, such as $SO(3)$-valued functional data \citep{Telschow2016},
recordings of densely sampled  trajectories on the sphere,  including flight trajectories \cp{anir:17, dai:17:1}, or
functions residing on unknown manifolds  \citep{Chen2012}. 

Since functional data are intrinsically infinite-dimensional, dimension reduction is a necessity, and a convenient and popular tool for this is  functional principal component
analysis, which is geared towards linear functional data and is not suitable for functional data on  nonlinear manifolds, 
for which  \citet{dai:17:1} investigated an intrinsic Riemannian  Functional Principal Component
Analysis (Riemannian FPCA) for functions taking values on a nonlinear Euclidean
submanifold, with  a Fr\'{e}chet type mean curve.  The concept of Fr\'{e}chet mean as a minimizer of the Fr\'echet function extends 
the classical mean in Euclidean spaces to data on Riemannian manifolds \cp{patra:18}.   Using Riemannian logarithm maps, data on manifolds can be  mapped
into tangent spaces identified with hyperplanes in the ambient space
of the manifold. Then  Riemannian FPCA   can be  conducted on the mapped data, where the 
Fr\'{e}chet mean and Riemannian logarithm maps reflect 
the curvature of the underlying manifold, yielding representations that are intrinsic to the manifold, which is an advantage over extrinsic approaches. 
 
A challenge is that  functional data are often sparsely observed, i.e.  each function is only recorded at an irregular grid
consisting of a few points. Such sparse recordings are routinely encountered in  longitudinal studies \cp{verb:14}.  
For example, in a longitudinal survey of unemployed workers in New Jersey \citep{krue:11} that we analyze in Section 4,  
the number of longitudinal responses available per subject is less than 4  in more than a half of the subjects. For sparsely observed longitudinal/functional
data such as these,  observations  need to be pooled across subjects in order
to obtain sensible estimates of mean  and covariance functions, as the data available for individual subjects are so sparse that meaningful smoothing is not possible.
 This pooling idea is at the core of 
the principal analysis by conditional expectation (PACE) approach \citep{Yao2005a}, whereas for densely sampled 
functional data one can apply individual curve smoothing or  cross-sectional 
strategies \citep{Zhang2007}.

An special case of longitudinal data are  longitudinal compositional data, 
where at each time point one observes fractions or percentages for each of a fixed number of categories, which add up to one. Such data occur in many applications, eg., repeated voting, with  counts transformed into percentages of votes for items, consumer preferences in terms of what fraction prefers a certain item, microbiome \cp{li:15},
online prediction markets, soil or air composition over time, mood assessment, and shape analysis, where we will study data of the latter two types of longitudinal data in Section~\ref{sec:Application}. 
While a classical approach for compositional data is to apply the simplex geometry in the form of the Aitchison geometry \cp{aitc:86} or a 
variant \cp{egoz:03,tals:18}, 
a disadvantage is that a baseline category needs to be identified, which  cannot have null outcomes, due to the need to form quotients; this is 
is especially difficult to satisfy in longitudinal studies, where null outcomes may fluctuate between categories.

 Motivated by the need to analyze sparsely sampled longitudinal data as for example found in the emotional well-being data collected in a longitudinal survey for unemployed workers in New Jersey and containing  a substantial proportion of null outcomes,
 we develop a Riemannian principal analysis method geared towards sparsely and irregularly observed  Riemannian
functional data. 

The main contributions of this paper are three-fold:

(1) We  develop a principal component analysis for longitudinal compositional data, which we also illustrate with  sparsely sampled body shape growth curves of Tammar wallabies, 
extending the scope of the approach of \cite{dai:17:1}, which only applies to  densely observed data.  To our knowledge, no methods exist yet
for the analysis of longitudinal data on manifolds. 

(2) We   extend  Fr\'echet regression \cite{mull:18:3} to functional data, while the approach in \cite{mull:18:3} was restricted to nonfunctional data as dependence between repeated measurements is not taken into account. 

(3) Concerning  theoretical analysis, we extend the techniques developed in \citep{Li2010,Zhang2016} to manifold-valued data and obtain 
rates of uniform convergence for the mean function.   The lack of a vector space structure makes this technically challenging. 

To obtain intrinsic representations of the unobserved trajectories on a nonlinear Riemannian manifold from sparsely observed longitudinal data, we first pool data from all subjects to obtain estimates for the mean and covariance function, and then obtain estimates of the individual principal components and trajectories by Best Linear Unbiased Prediction (BLUP). We employ a manifold local linear smoothing approach to estimate the Fr\'{e}chet mean curve, extending the approach of \citet{mull:18:3} for sparsely observed Riemannian functional data. Local linear smoothing was originally studied in the context of 
Euclidean non-functional data  \cp{fan:96} and
later has been extended to curved non-functional data \cp{Yuan2012}. 
Observations
of each function are then mapped into the tangent spaces around the
estimated mean curve via Riemannian logarithm maps. As the log-mapped
observations are vectors in the ambient space of the manifold, we
proceed by adopting a scatterplot smoothing approach to estimate
the covariance structure of the log-mapped data and then obtain
a finitely truncated representation, where the principal component
scores are estimated by  PACE, or sometimes integration, depending on
the sparseness of the observations available per function. Finally, a finite-dimensional
representation for the original data is obtained  by applying Riemannian exponential
maps that pull the log-mapped data back to the manifold. 

\bco

Details of the RPACE approach will be described in  Section \ref{sec:Methodology}, with theoretical properties presented in Section
\ref{sec:Theory}, where we show that the estimates of mean curve
and covariance function have the same  uniform convergence
rates as those obtained in  \citet{Zhang2016} for Euclidean functional data.
We applied RPACE to analyze the two motivating datasets,  longitudinal emotion composition for unemployed workers in New Jersey, and  body shapes for Tammar wallabies, both of which are sparse, with details provided   in Section~\ref{sec:Application}. We further studied the numerical properties and finite sample  behavior  of RPACE on $S^2$ and SO$(3)$ in Section \ref{sec:Simulation} via Monte Carlo simulations, where  RPACE was shown to exhibit  preferable trajectory recovery performance as compared with extrinsic methods under differing sample sizes and sparsity scenarios. 

\fi

\section{Methodology}\label{sec:Methodology}

\subsection{Statistical Model}\label{subsec:sm}

Let $\manifold$ be a $d$-dimensional, connected and geodesically
complete Riemannian submanifold of $\real^{D}$, where $d$ and $D$
are positive integers such that $d\leq D$. The dimension $d$ is
 the intrinsic dimension of the manifold $\manifold$, while
$D$ is the ambient dimension. The Riemannian metric $\innerprod{\cdot}{\cdot}$
on $\manifold$, which defines a scalar product $\innerprod{\cdot}{\cdot}_{p}$
for the tangent space $\tangentspace p$ at each point $p\in\manifold$,
is induced by the canonical inner product of $\real^{D}$, and it 
also induces a geodesic distance function $d_{\manifold}$ on $\manifold$.
A brief introduction to  Riemannian manifolds can be found in the appendix of \citet{dai:17:1}, see also 
 \cite{Lang1995} and \cite{Lee1997}.

We define a  
$\manifold$-valued
\textit{Riemannian random process}, or simply Riemannian random process
  $X(t)$, as a $D$-dimensional vector-valued random process defined
on a compact domain $\tdomain\subset\real$ such that $X(t)\in\manifold$, where we  assume that the process $X$ is of second-order, in the sense that
for every $t\in\tdomain,$ there exists $p\in\manifold$, potentially
depending on $t$,  such that the Fr\'{e}chet variance $M(p,t)\define\expect d_{\manifold}^{2}(p,X(t))$ 
is finite. For a fixed $t$, if $p$ is a point on $\manifold$ satisfying
$M(p,t)=\inf_{q\in\manifold}M(q,t)$, then $p$ is 
a Fr\'{e}chet mean  of $X$ at $t$.
Under conditions described in  \citet{Bhattacharya2003},  
the Fr\'{e}chet mean of a random variable on a manifold exists and
is unique, which we shall assume for $X(t)$ at all $t\in\tdomain$.

\begin{enumerate}[leftmargin=*,labelsep=7mm,label=(X\arabic*)]
\setcounter{enumi}{-1}
\item\label{cond:X0} 
$X$ is of second-order, and the Fr\'{e}chet mean curve $\mu(t)$ exists and is unique.
\end{enumerate}
Formally, we define the unique Fr\'{e}chet mean
function $\mu$ by 
\begin{equation}
\mu(t)=\underset{p\in\manifold}{\arg\min}\,M(p,t),\qquad t\in\tdomain.\label{eq:frechet-variances}
\end{equation}

As $\manifold$ is geodesically complete, by the Hopf--Rinow theorem, its
exponential map $\Exp_{p}$ at each $p$ is defined on the entire
$\tangentspace p$. To make $\Exp_{p}$ injective, define the domain $\mathscr{D}_{p}$ to be the interior of the collection of tangent vectors $v\in\tangentspace p$ such that if $\gamma(t)=\Exp_p(tv)$ is a geodesic emanating from $p$ with the direction $v$, then $\gamma([0,1])$ is a minimizing geodesic. 
Then on the domain $\mathscr{D}_{p}$  
the map $\Exp_{p}$ is injective, and its image is denoted by $\mathrm{Im}(\Exp_{p})$. 
The Riemannian logarithm map at $p$, denoted by $\Log_{p}$,
is the inverse of $\Exp_{p}$ restricted to $\mathrm{Im}(\Exp_{p})$. Specifically,
if $q=\Exp_{p}v$ for some $v\in\mathscr{D}_{p}$, then $\Log_{p}q=v$.
To study the covariance structure of the random process $X$
on tangent spaces, we will assume
\begin{enumerate}[leftmargin=*,labelsep=7mm,label=(X\arabic*)]
\item\label{cond:X1} For some constant $\epsilon_0>0$,  $\quad \prob\{X(t)\in \manifold \backslash (\manifold\backslash \mathrm{Im}(\Exp_{\mu(t)}))^{\epsilon_0} \,\, \text{for all}\,\,   t\in\tdomain\}=1$, where
$A^{\epsilon_0}$ denotes the set $\bigcup_{p\in A}\{q\in \manifold:d_{\manifold}(p,q)<\epsilon_0\}$.
\end{enumerate}
This condition requires $X(t)$ to stay away from the cut locus of $\mu(t)$ uniformly for all $t\in\tdomain$, which is necessary for the logarithm map $\Log_{\mu(t)}$ to be well defined, and is not needed  if $\Exp_{\mu(t)}$ is injective on $T_{\mu(t)}\manifold$ for
all $t$. In the special case of a  $d$-dimensional unit sphere $\mathbb{S}^{\dim}$,
if $X(t)$ is continuous and the distribution of $X(t)$ vanishes
at an open set with positive volume that contains $\manifold\backslash \mathrm{Im}(\Exp_{\mu(t)})$, (X1) holds.  
Under  \ref{cond:X0} and \ref{cond:X1}, $\Log_{\mu(t)}X(t)$
is almost surely defined for all $t\in\tdomain$. 
We will write  $L(t)$ to denote the  $\real^{D}$-valued random process $\Log_{\mu(t)}X(t)$ and refer to $L$ as the
log process of $X$.  

An important observation 
\cp{Bhattacharya2003} is that $\expect L(\cdot)\equiv0$. Furthermore,
the second-order condition on $X$ passes on to $L$, i.e., $\expect\|L(t)\|_{2}^{2}=\expect d^2_\manifold(\mu(t),X(t))<\infty$
for every $t\in\tdomain$, where $\|\cdot\|_{2}$ denotes the canonical
Euclidean norm in $\real^{D}$. This enables us to define the covariance
function of $L$ by 
\begin{equation}
\Gamma(s,t)=\expect\{L(s)L(t)^{\transpose}\},\qquad s,t\in\tdomain.\label{eq:cov-func-L}
\end{equation}
This covariance function admits the  eigendecomposition
\[
\Gamma(s,t)=\sum_{k=1}^{\infty}\lambda_{k}\phi_{k}(s)\phi_{k}^{\transpose}(t),
\]
where $\phi_{k}$ are orthonormal, $\lambda_{k}\geq\lambda_{k+1}$,
and $\sum_{k=1}^{\infty}\lambda_{k}<\infty$. The logarithm
process $L$ has the Karhunen-Lo\`eve expansion
\[
L(t)=\sum_{k=1}^{\infty}\xi_{k}\phi_{k}(t),
\]
where 
\begin{equation} \label{eq:xiDef}
\xi_{k} = \int_{\mathcal{T}} L(t)^{\transpose} \phi_k(t) dt
\end{equation}
are uncorrelated random variables such that $\expect\xi_{k}=0$
and $\expect\xi_{k}^{2}=\lambda_{k}$. 

A finite-truncated representation
of $X$ intrinsic to the manifold is then given by 
\begin{equation} \label{eq:XKLK}
X_{K}(t)\define\Exp_{\mu(t)}L_{K}(t), \quad L_{K}(t)=\sum_{k=1}^{K}\xi_{k}\phi_{k}(t)
\end{equation}
for some integer $K\geq0$. It was demonstrated in \citet{dai:17:1}  that this representation
is superior in terms of trajectory approximation for densely/completely observed manifold valued functional data compared to functional principal component analysis (FPCA), which is not adapted to the intrinsic manifold curvature, and for the same reason the 
scores $\xi_{k}$ are better predictors for classification tasks when compared
to traditional FPCs. 

Suppose $X_{1},\ldots,X_{n}$ are i.i.d. realizations of a $\manifold$-valued Riemannian
random process $X$. To reflect the situation in longitudinal studies, we  assume that each $X_{i}$ is only recorded
at $m_{i}$ random time points $T_{i,1},\ldots,T_{i,m_{i}}\in\tdomain$,
and each observation $X_{i}(T_{ij})$ is furthermore corrupted by some intrinsic
random noise. More specifically, we observe  $\mathbb{D}_{n}=\{(T_{ij},Y_{ij}):\:i=1,2,\ldots,n,\:j=1,2,\ldots,m_{i}\}$
such that $T_{i,j}\overset{i.i.d.}{\sim}f$ for some density $f$
supported on $\tdomain$, and the $T_{i,j}$ are independent of the $X_i$. Furthermore, conditional on $X_{i}$ and $T_{i,1},\ldots,T_{i,m_{i}}$,
the noisy observations $Y_{ij}=\Exp_{\mu(T_{ij})}\left\{ L_i(T_{ij})+\varepsilon_i(T_{ij})\right\} $
are independent, where $\varepsilon_i(T_{ij})\in\tangentspace{\mu(T_{ij})}$
is independent of $X_{i}$, with isotropic variance $\sigma^2$ and $\expect\{\varepsilon_i(T_{ij})\mid T_{ij}\}\equiv0$.
As  $\expect\{L_i(T_{ij})\mid T_{ij}\}\equiv0$, the assumption
on $\varepsilon$ implies that $\expect\{\Log_{\mu(T_{ij})}Y_{ij}\mid T_{ij}\}\equiv0$.

\subsection{Estimation } \label{ss:estimation}

For the case of sparse functional or longitudinal data that are the focus of this paper,  it is not possible to estimate the mean curve using the cross-sectional approach of \cite{dai:17:1}, as repeated observations at the same time $t$ are not available. Instead we develop a new method, for which we harness   
 Fr\'echet regression \citep{mull:18:3}. Fr\'echet regression was developed for independent measurements, and for our purposes we need to 
 study an extension that is valid for     
the case of repeated measurements.  

For any  $t\in\tdomain$ and $K_{h_{\mu}}$,
where the kernel $K(\cdot)$ is a symmetric density function and  $h_{\mu}>0$ is a sequence of bandwidths, with $K_{h_{\mu}}(x) =h_{\mu}^{-1}K(x/h_{\mu})$, we define the local weight
function at $t$ by
\[
\hat{\omega}_{ij}(t,h_{\mu})=\frac{1}{\hat{\sigma}_{0}^{2}}K_{h_{\mu}}(T_{ij}-t)\{\hat{u}_{2}-\hat{u}_{1}(T_{ij}-t)\},
\]
where $\hat{u}_{k}(t)=\sum_{i=1}^{n}w_{i}\sum_{j=1}^{m_{i}}K_{h_{\mu}}(T_{ij}-t)(T_{ij}-t)^{k}$
for $k=0,\,1,\,2$, and $\hat{\sigma}_{0}^{2}(t)=\hat{u}_{0}(t)\hat{u}_{2}(t)-\hat{u}_{1}^{2}(t)$.
Defining the double-weighted Fr\'echet function 
\[
Q_{n}(y,t)=\sum_{i=1}^{n}w_{i}\sum_{j=1}^{m_{i}}\hat{\omega}(T_{ij},t,h)d_{\manifold}^{2}(Y_{ij},y),
\]
which includes weights  $w_{i}$  for individual subjects satisfying $\sum_{i=1}^{n}m_{i}w_{i}=1$, we 
estimate the mean trajectory $\mu(t)$ by 
\[
\hat{\mu}(t)=\underset{y\in\manifold}{\arg\min}\,\,Q_{n}(y,t).
\]
Note that for the Euclidean special case, where  $\manifold=\real^{D}$,  $Q_{n}$ coincides
with the loss function used in \citet{Zhang2016} for linear functional
data. 

For the choice of the weights $w_i$, two  options have been studied in the Euclidean special case. One is to assign equal weight to each observation, i.e., $w_i=1/(n\bar{m})$ with $\bar{m}=n^{-1}\sum_{i=1}^n m_i$, used in \cite{Yao2005a}. The other is to assign equal weight to each subject, i.e., $w_i=1/(nm_i)$, as proposed in \ci{Li2010}. We refer to the former scheme as ``OBS'' and to the latter as ``SUBJ'', following  \cite{Zhang2016},  who found that the OBS scheme is generally preferrable for non-dense functional data; the SUBJ scheme performs better for ultra-dense data; and an intermediate weighting  scheme that is in between  OBS and SUBJ performs at least as well as the  OBS and SUBJ schemes in the Euclidean case. The latter corresponds to  the choice  $w_i=\alpha/(n\bar{m})+(1-\alpha)/(nm_i)$ for a constant $\alpha=c_2/(c_1+c_2)$ with $c_1=1/(\bar{m}h_\mu)+\bar{m}_2/\bar{m}^2$ and $c_2=1/(\bar{m}_Hh_\mu)+1$, where $\bar{m}_2=n^{-1}\sum_{i=1}^n m_i^2$ and $\bar{m}_H=n/\sum_{i=1}^nm_i^{-1}$, and we refer to this choice  as INTM. 

To estimate the covariance structure, we first map the original data
into tangent spaces, setting $\hat{L}_{ij}=\Log_{\hat{\mu}(T_{ij})}Y_{ij}$
and treating $\hat{L}_{ij}$ as a column vector in $\real^{D}$. To smooth  $D\times D$ matrices $\Gamma_{ijl}=\hat{L}_{ij}\hat{L}_{il}^{\transpose}$
for $j\neq l$, we extend the
scatterplot smoother \citep{Yao2005a} to matrix-valued data by finding minimizing $D\times D$ matrices $\hat{A}_0$, $\hat{A}_{1}$
and $\hat{A}_{2}$ according to
\be \la{wlse} 
&&(\hat{A}_0,\hat{A}_{1},\hat{A}_{2})\\
&&\nonumber \define \underset{A_{0},A_{1},A_{2}}{\arg\min}\sum_{i=1}^{n}v_{i}\sum_{1\leq j\neq l\leq m_{i}}\|\Gamma_{ijl}-A_0-(T_{ij}-s)A_{1}-(T_{il}-t)A_{2}\|_{F}^{2}K_{h_{\Gamma}}(T_{ij}-s)K_{h_{\Gamma}}(T_{il}-t),
\ee
where in the above weighted least squares error minimization step $\|\cdot\|_{F}$ is the matrix Frobenius norm, $h_{\Gamma}>0$
is a bandwidth, and $v_{i}$ are weights with $\sum_{i=1}^{n}m_{i}(m_{i}-1)v_{i}=1$. 

For the OBS weight scheme, $v_i=1/\sum_{i=1}^nm_i(m_i-1)$,  for the SUBJ scheme, $v_i=1/[nm_i(m_i-1)]$,  
while for INTM,   $v_i=\alpha/\sum_{i=1}^nm_i(m_i-1)+(1-\alpha)/[nm_i(m_i-1)]$ for a constant $\alpha=c_2/(c_1+c_2)$ with $c_1=1/(\bar{m}_2h_\Gamma^2)+\bar{m}_3/(\bar{m}_2^2h_\Gamma)+\bar{m}_4/\bar{m}_2^2$ and $c_2=1/(\bar{m}_{Q}h_\Gamma^2)+1/(\bar{m}_Hh_\Gamma)+1$, where $\bar{m}_k=n^{-1}\sum_{i=1}^n m_i^k$ and $\bar{m}_{Q}=n/\sum_{i=1}^nm_i^{-2}$, in analogy to \citet{Zhang2016}. We then
use $\hat{A}_0$  as obtained in  (\ref{wlse}) as an estimate  of the
population covariance function $\Gamma(s,t)$. For $s=t$, the minimization
is over symmetric matrices  $A_{1},A_{2}$ and symmetric semi-positive
definite matrices $A$. Estimates for the eigenfunctions $\phi_{k}$
and $\lambda_{k}$ of $\Gamma$ are then obtained by the corresponding eigenfunctions
$\hat{\phi}_{k}$ and eigenvalues $\hat{\lambda}_{k}$ of $\hat{\Gamma}$.


In applications, one needs  to choose appropriate bandwidths $h_{\mu}$ and
$h_{\Gamma}$, as well as the number of included components $K$. 
To select $h_{\mu}$ for smoothing the mean function $\mu$, we adopt a generalized cross-validation (GCV) criterion 
\[
\mathrm{GCV}(h)=\frac{\sum_{i=1}^{n}\sum_{j=1}^{m_{i}}d_{\manifold}^{2}(\hat{\mu}(T_{ij}),Y_{ij})}{(1 - K_{h}(0)/N)^2},
\]
where $N = \sum_{i=1}^n m_i$ is the total number of observations, and then choose $h_\mu$ as the minimizer of GCV$(h)$. 
While a similar GCV strategy can be adopted to select the bandwidth for the covariance function $\Gamma$, we propose to employ the simpler choice $h_\Gamma = 2h_\mu$, which we found to perform well numerically and which is computationally efficient.

To determine the number of components $K$ included in the finite-truncated representation \eqref{eq:XKLK}, it is sensible to consider the fraction of variation explained (FVE)
\begin{equation} \label{eq:FVE}
\text{FVE}(K) = \frac{\sum_{k=1}^K \lambda_k}{\sum_{j=1}^\infty \lambda_j}, \quad \widehat{\text{FVE}}(K) = \frac{\sum_{k=1}^K \hlambda_k}{\sum_{j=1}^\infty \hlambda_j},
\end{equation}
choosing the number of included components as the smallest $K$ such that the FVE exceeds a specified  threshold $0 < \gamma < 1$,
\begin{equation} \label{eq:Kstar}
K^* = \min\{ K: \text{FVE}(K) \ge \gamma \}, \quad \hat{K}^* = \min\{ K: \widehat{\text{FVE}}(K) \ge \gamma \},
\end{equation} 
where  common choices of $\gamma$ are 0.90, 0.95, and 0.99. 

\subsection{\label{ss:RPACE} Riemannian Functional Principal Component Analysis \newline Through Conditional Expectation}
The unobserved scores $\xi_{ik}$ need to be estimated from the discrete samples $\{(T_{ij}, X_{ij})\}_{j=1}^{m_i}$ or log-mapped samples $\{(T_{ij}, L_{ij})\}_{j=1}^{m_i}$. Approximating \eqref{eq:xiDef} by numerical integration is not feasible when  
the number of repeated measurements per curve is small, in analogy to the Euclidean case \citep{Yao2005a,krau:15}. 
We therefore propose  Riemannian Functional Principal Component Analysis Through Conditional Expectation (RPACE), generalizing the PACE procedure of \cite{Yao2005a} for tangent-vector valued processes, where we apply best linear unbiased predictors  (BLUP) to estimate the 
$\xi_{ik}$, obtaining the RFPC scores 
\begin{equation} \label{eq:xiTilde}
\txi_{ik} = \blup[\xi_{ik} \mid \bLi] = \lambda_k \bphi_{ik}^\transpose \bSigLi^{-1} \bLi.
\end{equation}

Here  $\blup$ denotes the best linear unbiased predictor. Writing $\Vec(\cdot)$ for the vectorization operation, $\bLi = \Vec([L_{i1}, \dots, L_{im_i}])$ are  the vectorized concatenated log-mapped observations for subject $i$, $\btLi = \Vec([L_i(T_{i1}), \dots, L_i(T_{im_i})])$, $\bphi_{ik} = \Vec([\phi_k(T_{i1}), \dots, \phi_k(T_{im_i})])$, and $\bSigLi = \expect(\bLi \bLi^T) = \expect(\btLi \btLi^T) + \sigma^2 I$, where  $I$ is the identity matrix. The entry of $\expect(\btLi \btLi^T)$ corresponding to $\expect([L_i(T_{ij})]_l [L_i(T_{il})]_m)$ is $[\Gamma(T_{ij}, T_{ik})]_{lm}$, where $[v]_{a}$ and $[A]_{ab}$ denote the $a$th or $(a,b)$th entry in a vector $v$ or  matrix $A$, respectively.  Substituting  corresponding estimates for the  unknown quantities in \eqref{eq:xiTilde},  we obtain plug-in estimates for $\xi_{ik}$, 
\begin{equation} \label{eq:xiHat}
\hxi_{ik} = \hlambda_k \hbphi_{ik}^\transpose \hbSigLi^{-1} \hbLi,
\end{equation}
where $\hbSigLi = \hat{\expect}(\btLi \btLi^T) + \hsigma^2I$; $\hat{\expect}(\btLi \btLi^T)$, $\hlambda_k$, and $\hbphi_{ik}$ are obtained from $\hat{\Gamma}$, the minimizer of (\ref{wlse}),   and we define $\hsigma^2 = \sum_{i=1}^n \sum_{j=1}^{m_i} (ndm_i)^{-1} \Tr(L_{ij}L_{ij}^T - \hat{\Gamma}(T_{ij}, T_{ij}))$, where $\Tr(A)$ denotes the trace of a matrix $A$. The $K$-truncated processes
\begin{equation} \label{eq:LKXK}
L_{iK}(t) = \sum_{k=1}^K \xi_{ik}\phi_k(t), \quad X_{iK}(t)= \sum_{k=1}^K \Exp_{\mu(t)}(L_{iK}(t))
\end{equation}
are estimated by 
\begin{equation} \label{eq:hLKXK}
\hL_{iK}(t) = \sum_{k=1}^K \hxi_{ik}\hphi_k(t), \quad \hX_{iK}(t) = \sum_{k=1}^K \Exp_{\hmu(t)}(\hL_{iK}(t)).
\end{equation}

The BLUP estimate $\txi_{ik}$ coincides with the conditional expectation $E[\xi_{ik}\mid \bLi]$, or the best prediction of $\xi_{ik}$ given observation $\bLi$, if the joint distribution of $(\xi_{ik}, \bLi)$ is elliptically contoured \citep[][Theorem 2.18]{fang:90}, with the Gaussian distribution as the most prominent example.  

\section{Asymptotic Properties}\label{sec:Theory}

To derive  the asymptotic properties of the estimates in Section \ref{sec:Methodology},
in addition to conditions \ref{cond:X0} and \ref{cond:X1}, we require 
the following assumptions.
\begin{enumerate}[leftmargin=*,labelsep=7mm,label= (M\arabic*)]
\setcounter{enumi}{-1}%
\item\label{cond:M0} The domain $\tdomain$ is compact and  the
manifold $\manifold$ is a bounded submanifold of $\real^{D}$.
\end{enumerate}

\begin{enumerate}[leftmargin=*,labelsep=7mm,label= (K\arabic*)]
\setcounter{enumi}{-1}%
\item\label{cond:K0} The kernel function $K$ is a Lipschitz continuous symmetric probability density
function on $[-1,1]$.
\end{enumerate}
\begin{enumerate}
[leftmargin=*,labelsep=7mm,label= (X\arabic*)]
\setcounter{enumi}{1}%
\item\label{cond:X2}Almost surely, the sample paths $X(\cdot)$ are
 twice continuously differentiable.
\end{enumerate}
Note that the boundedness assumption on the manifold can be relaxed
by imposing  additional conditions on the random process $X$, or by requiring  a compact support for $X(t)$, $t\in \tdomain$. The assumptions on the manifold are satisfied for our data applications in Section~\ref{sec:Application} where the manifolds under consideration are spheres.

To state the next assumption, we define the following quantities. Let
$\omega(s,t,h)=\frac{1}{\sigma_{0}^{2}}K_{h}(s-t)\{u_{2}(t)-u_{1}(t)(s-t)\}$,
where $u_{k}(t)=\expect\{K_{h}(T-t)(T-t)^{k}\}$, $k=0, 1, 2$, and $\sigma_{0}^{2}(t)=u_{0}(t)u_{2}(t)-u_{1}^{2}(t)>0$ for all $t$ by the Cauchy--Schwarz inequality. Note that the finiteness of $u_k$ is implied by the Lipschitz continuity of the kernel function $K$ and the compactness of the domain $\tdomain$. 
Define 
$\tilde{Q}_{h}(p,t)=\expect\{\omega(T,t,h)d_{\manifold}^{2}(Y,p)\}$
and $\tilde{\mu}(t)=\underset{y\in\manifold}{\arg\min}\,\tilde{Q}_{h}(y,t)$. 

\begin{enumerate}
[leftmargin=*,labelsep=7mm,label= (L\arabic*)]
\setcounter{enumi}{-1}%
\item\label{cond:L0} The Fr\'echet mean functions $\mu$, $\tilde{\mu}$, and $\hat{\mu}$
exist and are unique, the latter  almost surely for all $n$. 
\item\label{cond:L1} The density $f(t)$ of the random times $T$ when measurements are made is positive and twice continuously differentiable for $t\in\tdomain$.
\end{enumerate}

Recall that $\tangentspace p$ denotes the tangent space at $p\in\manifold$ and $\Exp_{p}$ is the Riemannian exponential map at $p$, which maps a tangent vector $v\in \tangentspace p$ onto the manifold $\manifold$. For $p\in\manifold$,  define a real-valued function $G_{p}(v,t)=M(\Exp_{p}v,t)$, $v \in \tangentspace p$ and $t \in \tdomain$, where $M(p,t)=\expect d_{\manifold}^{2}(p,X(t))$ is the Frech\'et variance function defined in Section \ref{subsec:sm}. We assume 
\begin{enumerate}
[leftmargin=*,labelsep=7mm,label= (L\arabic*)]
\setcounter{enumi}{1}%
\item\label{cond:L2} The Hessian of $G_p(\cdot,t)$ at $v=0$ is uniformly positive definite along the mean function, i.e.,
\[
\inf_{t\in\tdomain}\lambda_{\min}\left(\frac{\partial^{2}}{\partial v^{2}}G_{\mu(t)}(v,t)\mid_{v=0}\right)>0. 
\]
\end{enumerate}

Conditions \ref{cond:L0} is necessary to ensure a consistent estimate of the mean curve using $M$-estimation theory, while \ref{cond:L1} is a design density condition; both are standard in the literature  \citep{Zhang2016,mull:18:3}. 
On a Riemannian manifold $\manifold$ with sectional curvature at most $\mathcal{K}$, \ref{cond:L0} and \ref{cond:L2} are satisfied if the support of $X(t)$ is within $B(\mu(t), \pi/(2\mathcal{K}))$, where $B(p,r)$ is a geodesic ball with center $p\subset \manifold$ and radius $r$ \citep{bhat:12}; this specifically  holds for longitudinal compositional data mapped to the positive orthant of a unit sphere. The next two conditions
impose certain convergence rates for  $h_{\mu}$ and $h_{\Gamma}$, respectively. For simplicity, we shall assume $m_i\equiv m$, noting that  results paralleling  those in \cite{Zhang2016} can be obtained for the general case. 
\begin{enumerate}
[leftmargin=*,labelsep=7mm,label= (H\arabic*)]
\item\label{cond:H1} $h_{\mu}\rightarrow0$ and $(\log n)/(nmh_\mu)\rightarrow0$.
\item\label{cond:H2} $h_{\Gamma}\rightarrow0$, $(\log n)/(nm^2h_\Gamma^2)\rightarrow0$,
 and $(\log n)/(nmh_\Gamma)\rightarrow0$.
\end{enumerate}
The following result establishes the uniform convergence rate for
estimates  $\hat{\mu}$. 
\begin{thm}
\label{thm:mu-asymptotic}Assume conditions \ref{cond:X0}--\ref{cond:X2}, \ref{cond:M0},
\ref{cond:K0}, \ref{cond:L0}--\ref{cond:L2}  and
\ref{cond:H1} hold. Then
\begin{equation}
\sup_{t\in\tdomain}d_{\manifold}^{2}(\hat{\mu}(t),\mu(t))=\Op\left(h_{\mu}^{4}+\frac{\log n}{nmh_{\mu}}+\frac{\log n}{n}\right).\label{eq:mean-rate}
\end{equation}
\end{thm}

This result shows that the estimate $\hat{\mu}$ enjoys the same rate as the one obtained in  \citet{Zhang2016} for the Euclidean case, even in the presence
of curvature. The rate in \eqref{eq:mean-rate} has three terms that correspond to  three regimes that are characterized by the growth rate of
$m$ relative to the sample size: (1) When $m\ll(n/\log n)^{1/4}$,
the observations per curve are sparse, and the optimal choice $h_{\mu}\asympeq(nm/\log n)^{-1/5}$
yields $\sup_{t\in\tdomain}d_{\manifold}(\hat{\mu}(t),\mu(t))=\Op\left((nm/\log n)^{-2/5}\right)$;
(2) When $m\asympeq(n/\log n)^{1/4}$, corresponding to an  intermediate case, the optimal choice 
$h_{\mu}\asympeq(n/\log n)^{-1/4}$ leads to the uniform rate 
$\Op\left(\{(\log n)/n\}^{1/2}\right)$  for $\hat{\mu}$; (3) When $m\gg(n/\log n)^{1/4}$,
the observations are dense, and any choice $h_{\mu}=o\left((n/\log n)^{-1/4}\right)$
gives rise to the uniform rate  $\Op\left(\{(\log n)/n\}^{1/2}\right)$. The transition from (1) to (3) is akin to a phase transition, similar to the one observed in 
\ci{Hall2006a}. 

The next result concerns the uniform rate for the estimator $\hat{\Gamma}$ of $\Gamma$, the covariance function of the log-mapped data, extending a result of 
 \citet{Zhang2016} for the Euclidean case to curved functional data.
\begin{thm}
\label{thm:covariance-rate}Assume conditions \ref{cond:X0}--\ref{cond:X2},
\ref{cond:M0}, \ref{cond:K0}, \ref{cond:L0}--\ref{cond:L2}, 
\ref{cond:H1} and \ref{cond:H2} hold. Then
\begin{equation}
\sup_{s,t\in\tdomain}\|\hat{\Gamma}(s,t)-\Gamma(s,t)\|_{F}^{2}=\Op\left(h_{\mu}^{4}+h_{\Gamma}^{4}+\frac{\log n}{nmh_{\mu}}+\frac{\log n}{n}+\frac{\log n}{nm^{2}h_{\Gamma}^{2}}+\frac{\log n}{nmh_{\Gamma}}\right).
\label{eq:rate-cov}
\end{equation}
\end{thm}
Again, the  above rate gives rise to three regimes that are determined by the growth rate
of $m$ relative to the sample size: (1) When $m\ll(n/\log n)^{1/4}$,
the observations per curve are sparse, and with the optimal choice
$h_{\mu}\asympeq(nm/\log n)^{-1/5}$ and $h_{\Gamma}\asympeq(nm^{2}/\log n)^{-1/6}$,
one has $\sup_{s,t\in\tdomain}\|\hat{\Gamma}(s,t)-\Gamma(s,t)\|_{F}=\Op\left((nm^{2}/\log n)^{-1/3}\right)$;
(2) When $m\asympeq(n/\log n)^{1/4}$, with the optimal choice $h_{\mu}\asympeq h_{\Gamma}\asympeq(n/\log n)^{-1/4}$,
the uniform rate for $\hat{\Gamma}$ is $\Op\left(\{(\log n)/n\}^{1/2}\right)$;
(3) When $m\gg(n/\log n)^{1/4}$, the observations are dense, and
any choice $h_{\mu},h_{\Gamma}=o\left((n/\log n)^{-1/4}\right)$ yields
the uniform rate $\Op\left(\{(\log n)/n\}^{1/2}\right)$.  

 Furthermore, according
to Lemma 4.2 of \citet{Bosq2000}, one has $\sup_k|\hat{\lambda}_k-\lambda_k|\leq\| \hat{\Gamma}-\Gamma\|_{HS}$. It can also be shown that $\|\hat{\Gamma}-\Gamma\|_{HS}\leq |\tdomain|\sup_{s,t\in\tdomain}\|\hat{\Gamma}(s,t)-\Gamma(s,t)\|_{F}$, where $|\tdomain|$ denotes the Lebesgue measure of $\tdomain$. Therefore, the rate for $\hat{\Gamma}$
 provides a convergence rate for all estimated eigenvalues $\hat\lambda_k$. Furthermore, according to Lemma 4.3 of \citet{Bosq2000}, if $\lambda_{k-1}\neq\lambda_{k}$
and $\lambda_{k}\neq\lambda_{k+1}$, then $\|\hat{\phi}_{k}-\phi_{k}\|_{2}^{2}\leq c_{k}\|\hat{\Gamma}-\Gamma\|_{HS}^{2}$, where $c_1=8(\lambda_1-\lambda_2)^{-2}$ and $c_k=8\max\{(\lambda_{k-1}-\lambda_k)^{-2},(\lambda_k-\lambda_{k+1})^{-2}\}$ for $k\geq 2$. Again, by utilizing the fact that $\|\hat{\Gamma}-\Gamma\|_{HS}\leq |\tdomain|\sup_{s,t\in\tdomain}\|\hat{\Gamma}(s,t)-\Gamma(s,t)\|_{F}$, we can derive the convergence rate for $\hat{\phi}_k$. For example, if we assume polynomial decay of  eigenvalue   spacing, i.e., $a_1k^{-\beta}\leq\lambda_k-\lambda_{k+1}\leq a_2k^{-\beta}$ for some constants $a_2\geq a_1>0$ and $\beta>1$, then $\|\hat{\phi}_{k}-\phi_{k}\|_{2}^{2}=\Op(k^{2\beta}\gamma_n)$ where $\gamma_n$ is the rate that appears on the right hand side of \eqref{eq:rate-cov}, and the $\Op$ term is uniform for all $k$.

\section{Data Applications}\label{sec:Application}

\subsection{Emotional Well-Being for Unemployed Workers}
We demonstrate RPACE for the analysis of   longitudinal mood compositional data. 
These data were collected in the Survey of Unemployed Workers in New Jersey \citep{krue:11}, conducted in the fall of 2009 and the beginning of 2010, during which the unemployment rate in the US peaked at 10\% after the financial crisis of 2007--2008. A stratified random sample of unemployed workers were surveyed weekly  for up to 12 weeks. Questionnaires included an entry survey, which assessed demographics, household characteristics and income, and weekly followups, including job search activities and emotional well-being. In each followup questionnaire, participants were asked to report the percentage of time they spent in each of the four moods: bad, low/irritable, mildly pleasant, and good. The overall weekly response rate was around 40\%; see  \cite{krue:11}. 

We analyzed a sample of $n=4771$ unemployed workers enrolled in the study, who were not offered a job during the survey period. The measurement of interest $
Y(t) = [{Y_1(t)}, \dots, {Y_4(t)}]
$
is the longitudinal mood composition, where $Y_j(t)$ is the proportion of time a subject spent in the $j$th mood in the previous 7 days, $j=1, \dots, 4$, recorded on day $t\in[0, 84]$ since the start of the study. The number of responses per subject ranged from 1 to 12, so the data is a mixture of very  sparse and mildly sparse longitudinal observations; for 25\% of all subjects only one response was recorded.  As subjects responded at different days of the week, the observation time points were also irregular. The sparsity and irregularity of the observations poses difficulties for classical analyses and prevents the application of  
the presmooth-and-then-analyze method \citep{dai:17:1}, motivating the application of RPACE, which is geared towards such sparse and irregularly sampled manifold-valued functional data.  

We applied RPACE for the square-root transformed compositional data
\[
X(t) = \{X_1(t), \dots, X_4(t)\} = \{\sqrt{Y_1(t)}, \dots, \sqrt{Y_4(t)}\},
\]
which lie on the sphere $\mathbb{S}^3$ for $t \in [0, 84]$,  since compositional data are non-negative and sum to 1, using bandwidths $h_\mu = 18$ and $h_\Gamma = 36$ days, as selected by GCV,  and the Epanechnikov kernel $K(x)=0.75(1-x^2)$ on $[-1,1]$. The mood composition trajectories for four randomly selected subjects are displayed in the left panel of \autoref{fig:meanFittedNJUI}. The solid dots denote the reported moods, which are slightly jittered vertically if they  overlap, and dashed curves denote the fitted trajectories when selecting  $K=8$ components, selected according to the FVE criterion \eqref{eq:Kstar} with threshold $\gamma=0.99$, which is a reasonable choice in view of the  large sample size. A substantial proportion of the mood compositions is zero, which is no problem for the square-root transformation approach in contrast to the alternative log-ratio transformation 
\citep{aitc:86}, which is undefined when the baseline category is 0. 

As the self-reported moods contain substantial aberrations from smooth trajectories that we view as noise, the fitted trajectories do not go through the raw observations, and are drawn towards the observations for subjects with more repeated measurements. The mean trajectory is displayed in the right panel of \autoref{fig:meanFittedNJUI}, indicating  that the emotional well-being of subjects tends to deteriorate as the period of unemployment lengthens, with an overall increase in the  proportion of bad mood and a decrease in the proportion of good mood. 
 
\begin{figure}
\includegraphics[width=0.55\linewidth]{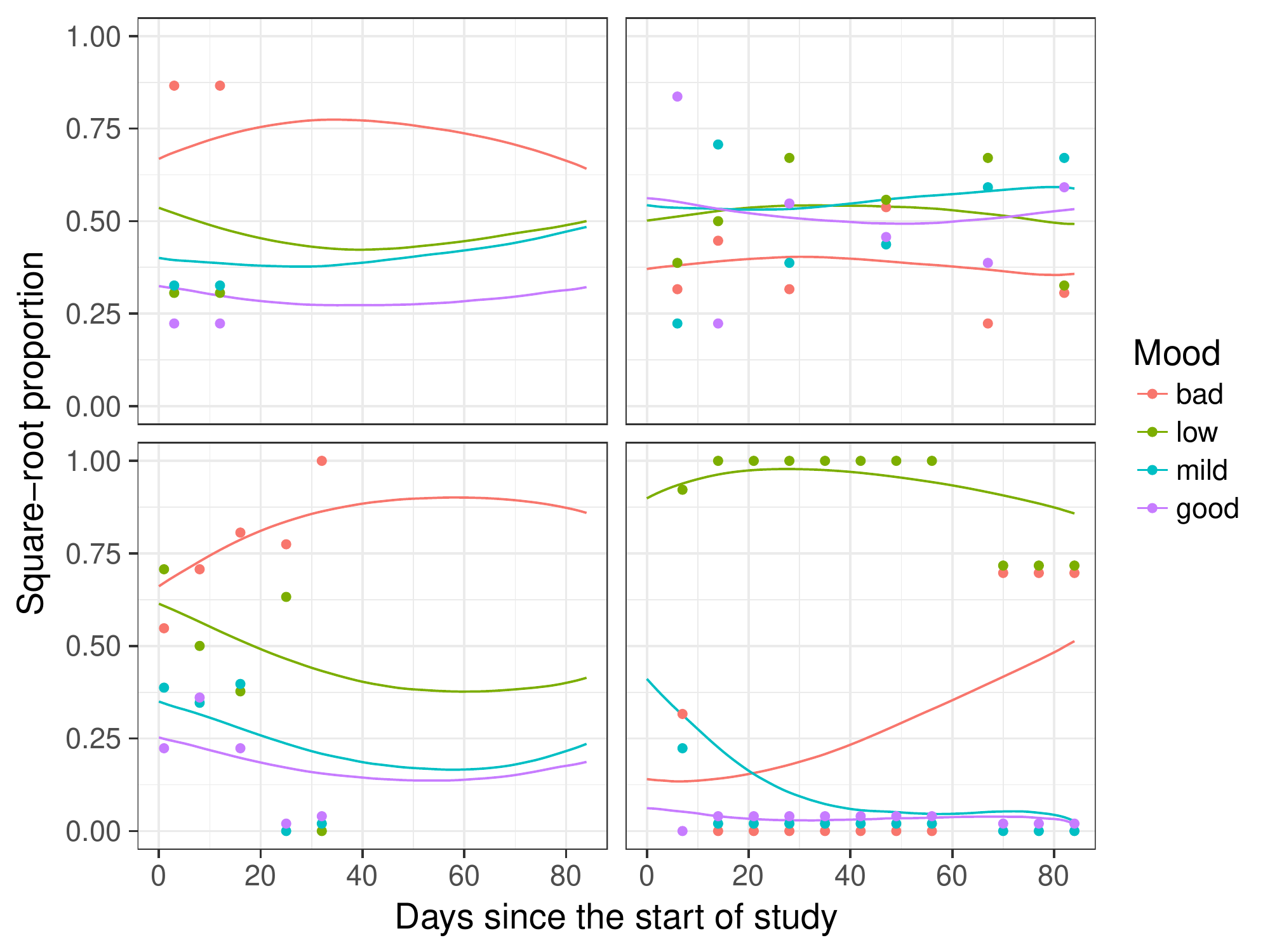}
\includegraphics[width=0.44\linewidth]{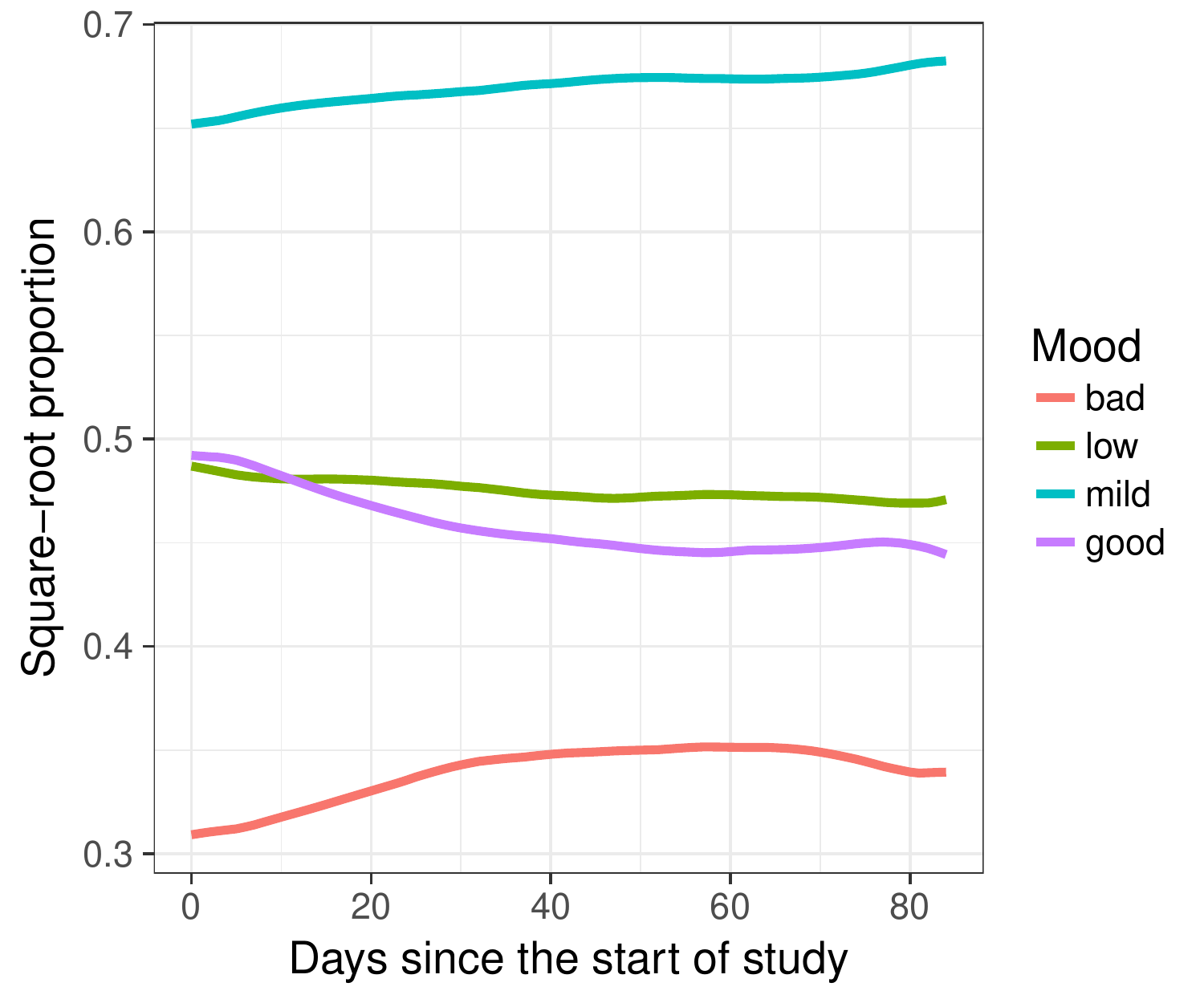}
\caption{Left: Longitudinal mood compositional data  for four randomly selected unemployed workers, with raw observations shown as dots and fitted trajectories by RPACE using   8 eigen-components shown as  solid curves. Overlapping dots  were slightly jittered vertically.  Right: The overall mean function.}
\label{fig:meanFittedNJUI}
\end{figure}

The first four eigenfunctions for mood composition trajectories are shown in \autoref{fig:phiNJUI}, where the  first eigenfunction corresponds to the overall contrast between neutral-to-positive mood (good and mild) and negative moods (low and bad); the second eigenfunction represents emotional stability, which is a contrast between  more neutral moods and extreme emotions (good and bad); the third eigenfunction corresponds to a shift of mood compositions to more positive moods, namely from bad to low and from mild to good; the fourth eigenfunction encodes an increase of positive feelings and a decrease of negative ones over time. Here it is important to note that the sign of the eigenfunctions is arbitrary and could be reversed. The first four eigenfunctions together explain 95\% of the total variation. 

\begin{figure}
\centering
\includegraphics[width=.8\linewidth]{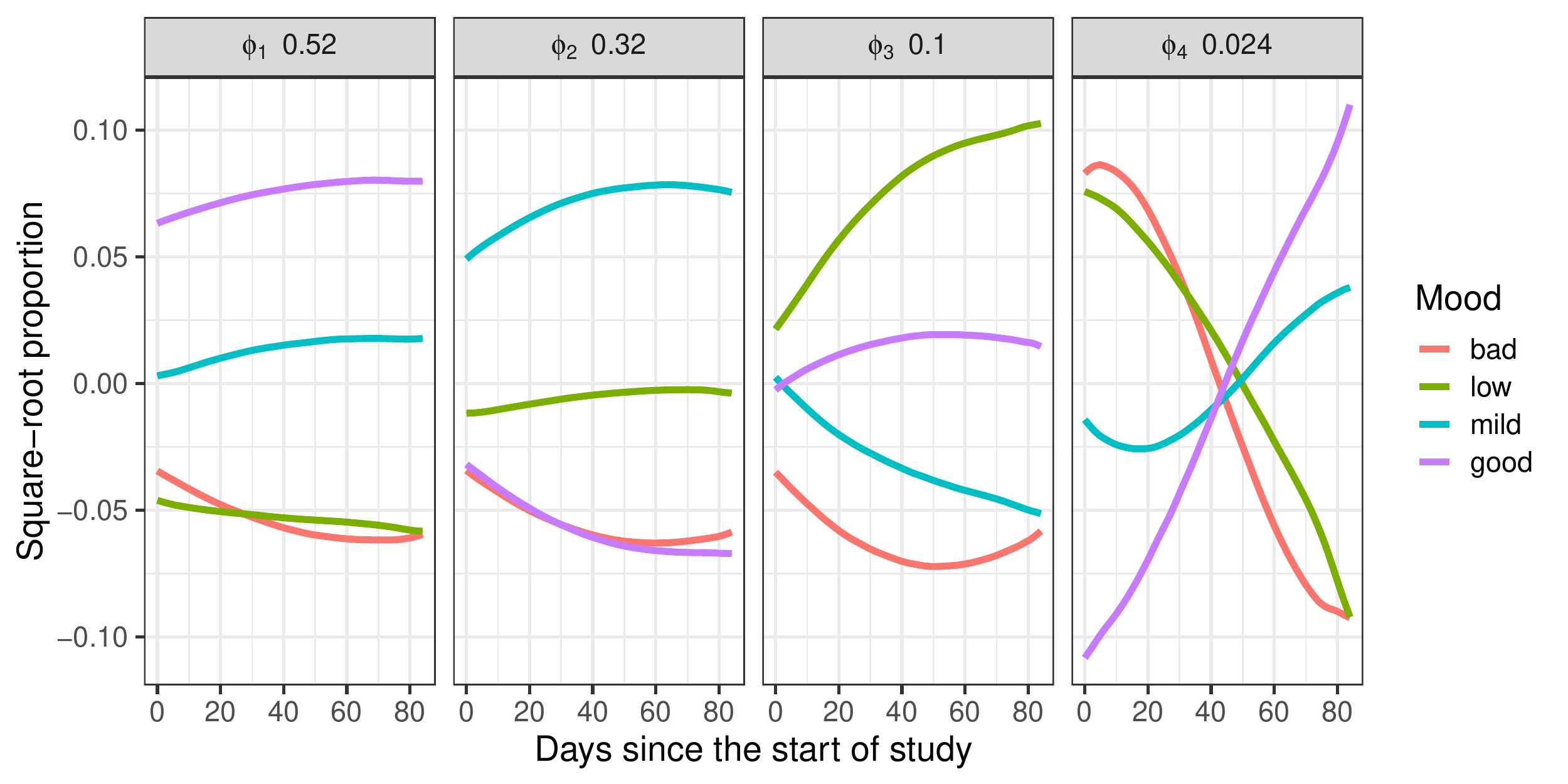}
\caption{The first four eigenfunctions for the mood composition data, with Fraction of Variation Explained (FVE) displayed in the panel subtitles. }
\label{fig:phiNJUI}
\end{figure}

As an example to demonstrate that the scores obtained from RFPC are useful for downstream tasks such as regression, we explored the association between the second RFPC score, corresponding to the proportion of extreme moods, and annual household income in 2008, a measure of financial stability. We extracted the RFPC scores for each subject and constructed kernel density estimates  for $\xi_2$ within each income category; see  \autoref{fig:xi2NJUI}. Participants with higher household income before losing their job and thus higher financial stability tend to have higher emotion stability, as demonstrated by the right-shifted distributions of $\xi_2$ and larger means (colored dots). The relationship between prior income and emotional stability appears to be nonlinear especially for the lower income groups.

\begin{figure}
\centering
\includegraphics[width=.8\linewidth]{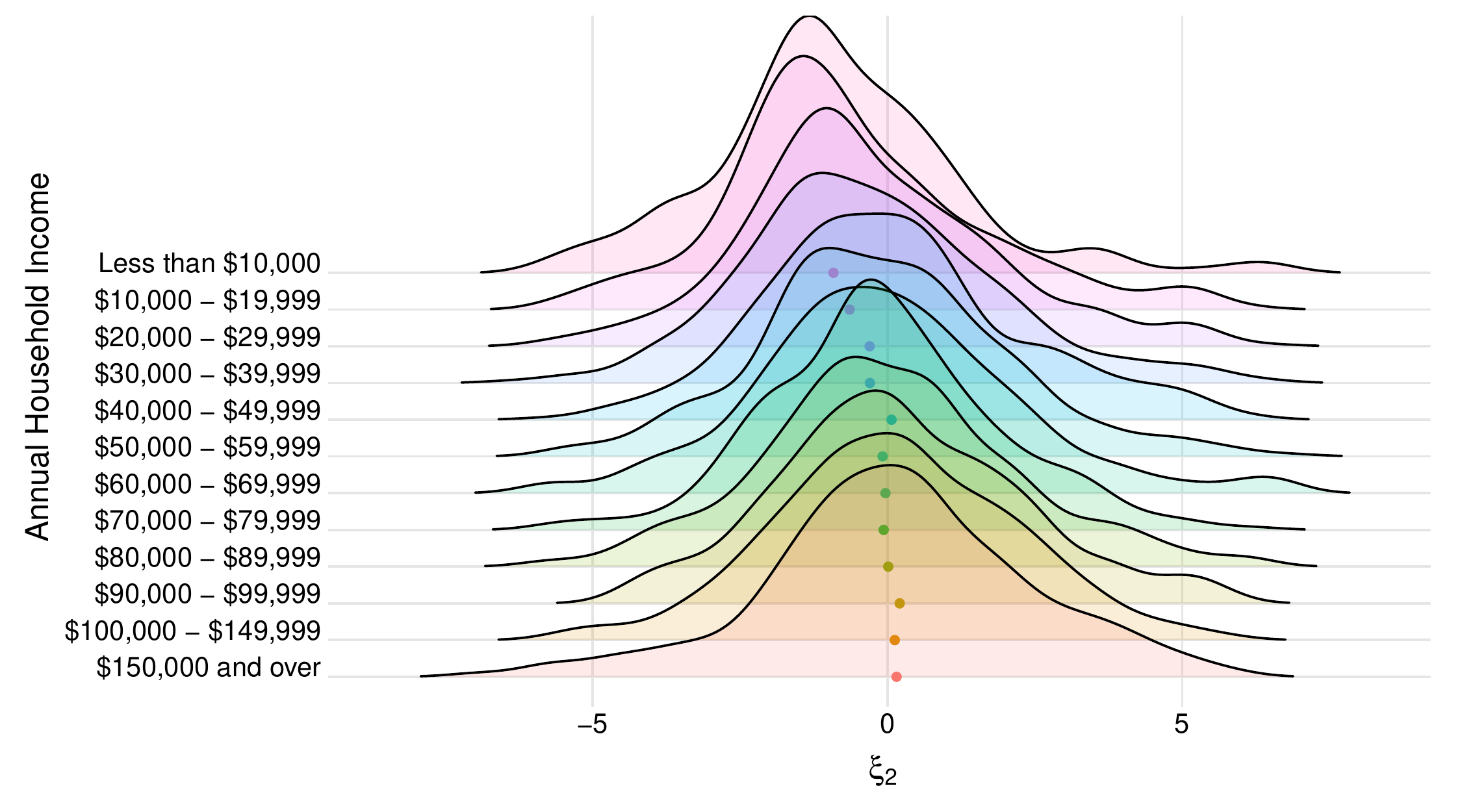}
\caption{The distributions of the second RFPC score, encoding emotion stability, visualized as densities in dependence on the annual household income in 2008. Colored dots indicate the mean of $\xi_2$ for each income group. }
\label{fig:xi2NJUI}
\end{figure}





\subsection{Wallaby Body Shape Growth} \label{ssec:wallaby}
Quantifying the  shapes  of organisms has been a long-standing statistical and mathematical problem  \citep{thom:42,kend:09}. We apply RPACE to analyze the  longitudinal development of body shapes of a sample of Tammar wallabies (\emph{Macropus eugenii}), a small macropod native to Australia (data courtesy of Dr Jeff Wood, CSIRO Biometrics Unit INRE, Canberra, and data cleaning and corrections were performed by Professor Heike Hofmann, Department of Statistics, Iowa State University, in 2008). 
For each of $n=40$ measured wallabies from two locations, longitudinal measurements of the length (in inches) of six body parts $Y_{ij} = (\text{Head}, \text{Ear}, \text{Arm}, \text{Leg}, \text{Foot}, \text{Tail})_{ij}$ were available at age $T_{ij}$ in the first 380 days after birth, for $i=1,\dots, 40$ and $j=1, \dots, n_i$. The measurement time points for the wallabies were highly irregular, and the number of measurements per wallaby varied from 1 to 26, with 14 wallabies having no more than 7 measurements. Typical measurement patterns with mixed sparse and dense observations for each curve are shown in the left panel of  \autoref{fig:meanFittedWallaby}. This   measurement scheme requires methodology that can handle the high degreee of irregularity in the measurement times.

To quantify  shapes of wallabies, we normalized the length measurements $Y_{ij}$ by the Euclidean norm, obtaining $X_{ij} = Y_{ij} / \|Y_{ij}\|_2$, thus emphasizing the relative size of each body part expressed as a percentage of total size, 
leading to longitudinal preshape data \citep{kend:09} that lie on a sphere. The  $X_{ij}$ are  shape characteristics of  wallabies at their respective age. 
RPACE was then applied to the transformed data $(T_{ij}, X_{ij})$ with bandwidths $h_\mu = 18.3$ and $h_\Gamma = 36.6$ selected by GCV, using  the Epanechnikov kernel. The Fr\'echet mean trajectory as displayed in the right panel of \autoref{fig:meanFittedWallaby} shows that  relative to the body size, the tail becomes larger,  while head,  arm and ear   lengths become relatively smaller  throughout the first year of birth; leg and foot lengths increase from birth to roughly 6 months, where relative leg length development peaks before that of foot development. 

\begin{figure}
\centering
\includegraphics[width=0.59\linewidth]{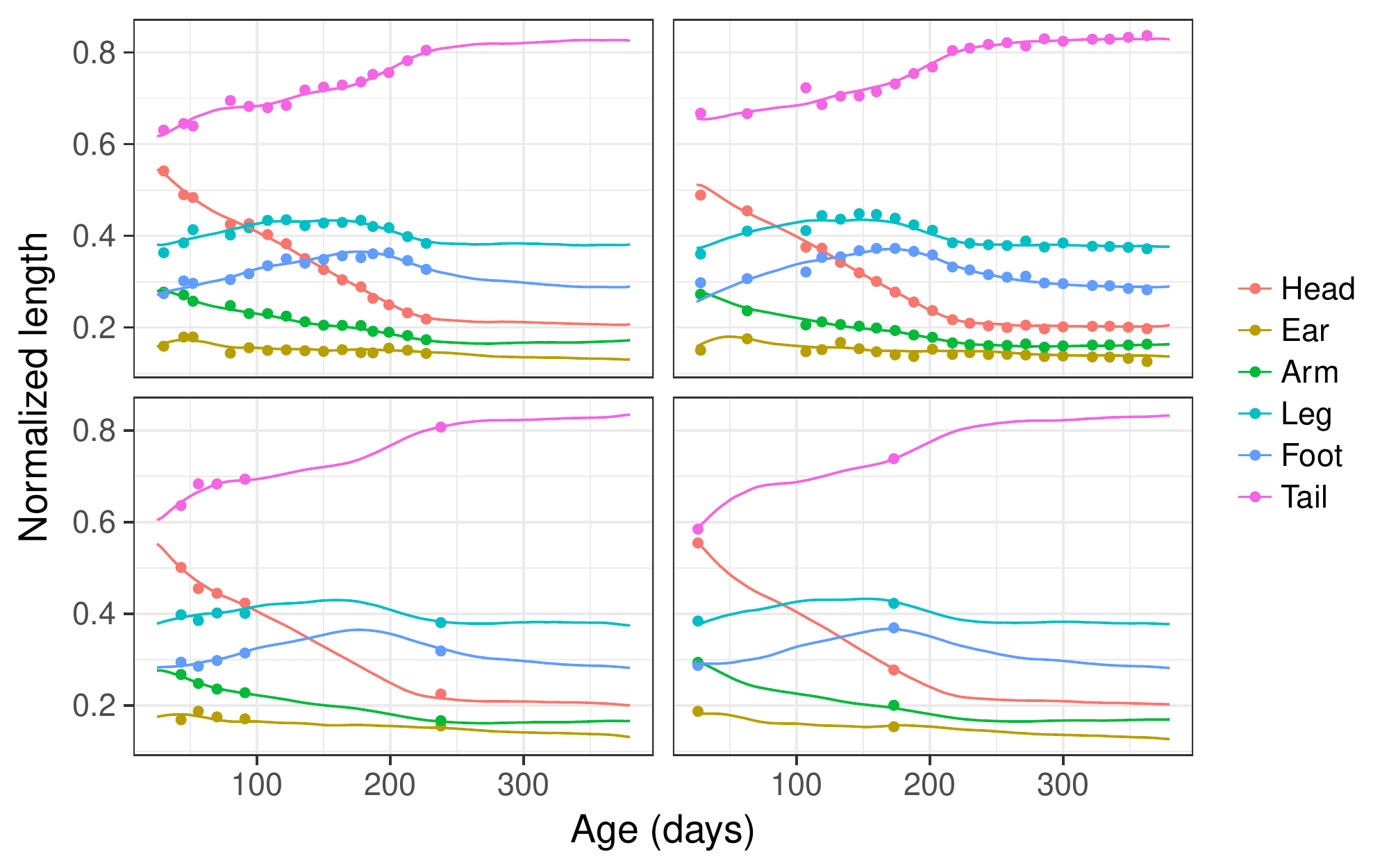}
\includegraphics[width=0.39\linewidth]{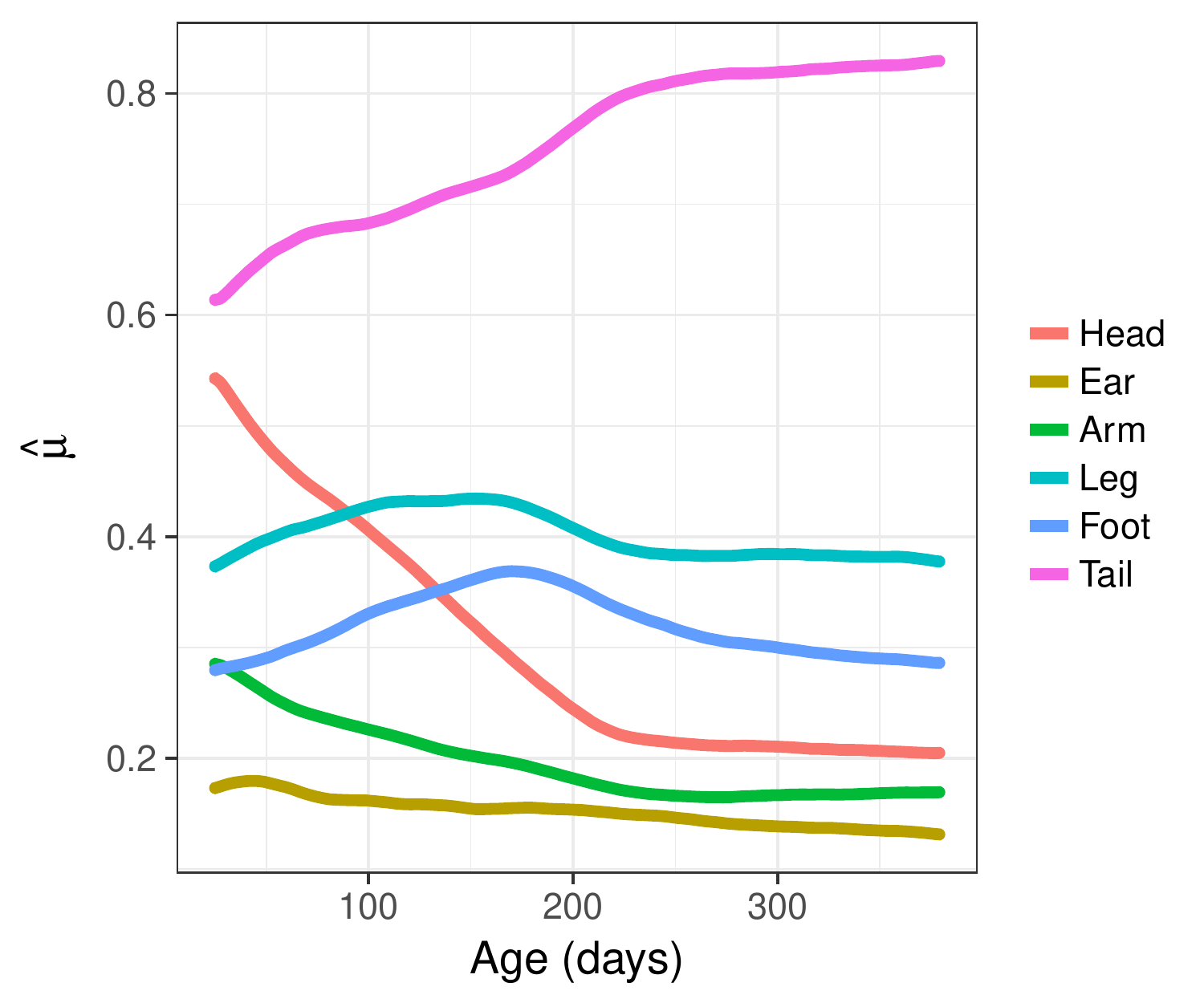}
\caption{Left: Longitudinal shape observations of four randomly selected wallabies, where dots denote raw observations, and solid curves indicates fitted trajectories by 7 components that together explain 90\% of total variation. Right: The mean function of all trajectories.}
\label{fig:meanFittedWallaby}
\end{figure}

To decompose the variation of individual shape trajectories, the first three eigenfunctions are displayed in \autoref{fig:phiWallaby}. The first eigenfunction corresponds to an overall contrast between tail and other body part development, and the second eigenfunction has a large component in the initial tail growth. The first two eigenfunctions together explain 64\% of total variation, showing that  tail length is a main driving force for shape differences. Pairwise scatterplots of the first three RFPC scores are shown in \autoref{fig:xiWallaby}, where each point stands for a single wallaby, and their patterns  indicate two different geographic locations. Shape development differences between locations were mostly reflected in the second component, corresponding to initial tail growth, while the first and third components were less dissimilar. 

\section{Simulation Studies}\label{sec:Simulation}
We demonstrate the performance of the sparse RPACE method for scenarios with varying sample size, sparsity, and manifolds, for which we choose $\manifold = S^2$ or SO$(3)$. Here $S^2$ is the 2-sphere and SO$(3)$  is the manifold consisting of  the $3 \times 3$ orthogonal matrices with determinant 1. 
For each random trajectory $X_i(t)$ on $\manifold$, $i=1,\dots, n$, we sample  $m_i$ observations $(T_{ij}, X_{ij})$, $j=1,\dots, m_i$, where $T_{ij}$ follows a uniform distribution on $\tdomain = [0, 1]$. The number of observations $m_i$ follows a discrete uniform distribution on $\{1, \dots, \mmax\}$, where $\mmax$ is the maximum number of observations per curve that differs between scenarios. 

The sparse observations were generated according to $X_{ij} = \Exp_{\mu(T_{ij})}(L_i(T_{ij})+ \epsilon_{ij} )$, $L_i(t) = \sum_{k=1}^{20}\xi_{ik} \phi_k(t)$, with manifold-specific mean function $\mu(t)$ and eigenfunctions $\phi_j(t)$; RFPC scores $\xi_{ik}$ that follow independent Gaussian distributions with mean zero and variance $\lambda_k=0.05^{k/3}$, for $k = 1, \dots, 20$; and independent Gaussian errors $\epsilon_{ij}$ with mean $0$ and isotropic variance $\sigma^2=0.01$ on the tangent space $T_{\mu(T_{ij})}$. The cumulative FVE for the first $K=1,\dots, 6$ components, defined as $\sum_{j=1}^K \lambda_j/ \sum_{k=1}^\infty \lambda_k$, are 63.2\%, 86.4\%, 95.0\%, 98.15\%, 99.3\%, and 99.8\%, respectively.
For $\manifold=S^2$, we set $\mu(t) = \Exp_{p}(\nu(t))$ where $p=[0, 0, 1]$ and $\nu(t) = [2t/2^{1/2}, 0.3\pi\sin(\pi t), 0]$; eigenfunctions $ 2^{-1/2} R_t [\zeta_k(t / 2), \zeta_k((t + 1)/2), 0]^T$, with $R_t$ being the rotation matrix from $p$  to $\mu(t)$, and $\{\zeta_k\}_{k=1}^{20}$ the orthonormal cosine basis on $[0,1]$. 
For $\manifold=\text{SO}(3)$, $\mu(t) = \expm(\iota(2t, 0.3\pi\sin(\pi t), 0))$ and $\phi_k(t) = 3^{-1/2}\iota(\zeta_k(t / 3), \zeta_k((t + 1)/ 3), \zeta_k((t+2)/3))$, where $\expm$ is the matrix exponential and $\iota: \mathbb{R}^3\rightarrow \mathbb{R}^{3\times 3}$ maps a vector $v$ to a skew-symmetric matrix whose lower diagonal elements (ordered by column) are $v$.
We investigated three settings with varying sparsity and sample size: 
Scenario 1 (baseline): $n = 100$, $\mmax = 20$;
Scenario 2 (sparse): $n = 100$, $\mmax = 5$;
Scenario 3 (small $n$): $n = 50$, $\mmax = 20$.

\begin{figure}
\centering
\includegraphics[width=.8\linewidth]{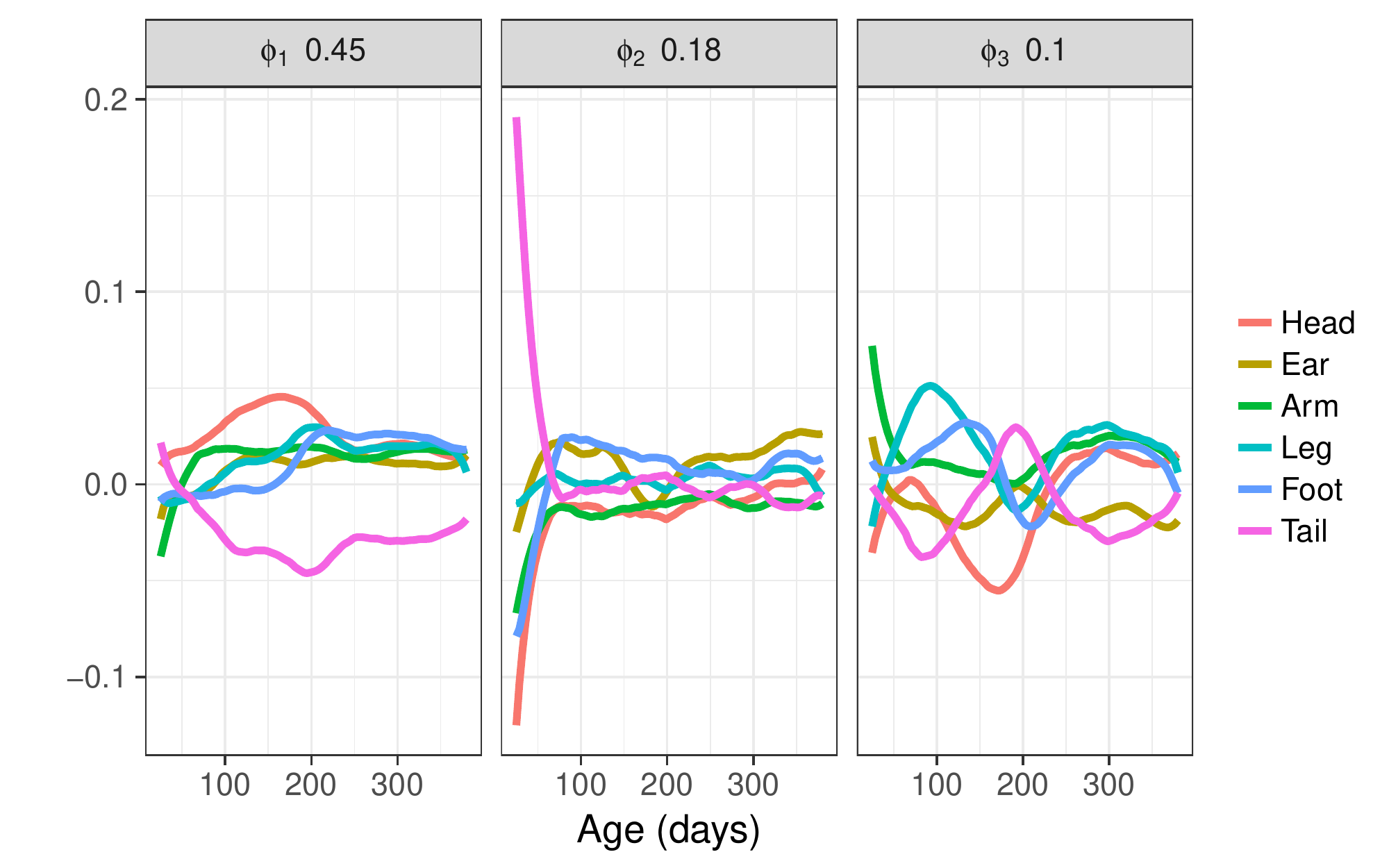}
\caption{The first four eigenfunctions for the wallaby shapes, with Fraction of Variation Explained (FVE) displayed in the panel subtitles. }
\label{fig:phiWallaby}
\end{figure}

\begin{figure}
\centering
\includegraphics[width=0.65\linewidth]{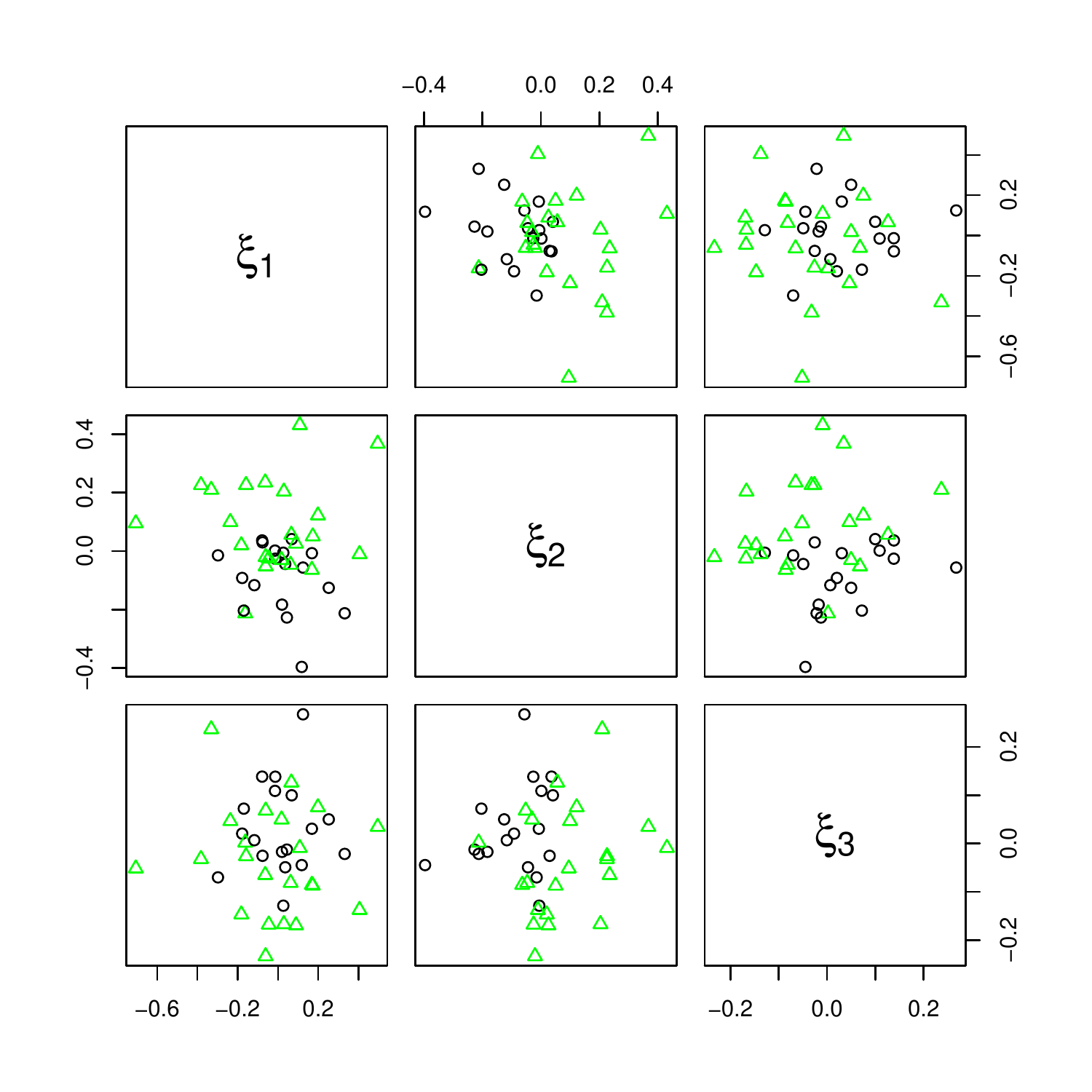}
\caption{Pairwise scatter plots of the first three RFPC scores for the Wallaby data, where different point patterns represent Wallabies from  two distinct geographic  locations. }
\label{fig:xiWallaby}
\end{figure}

Three different FPCA approaches were evaluated for these scenarios, namely an extrinsic FPCA by  \cite{chio:14}, an extrinsic multivariate FPCA via componentwise FPCA (CFPCA) by \cite{happ:18},  and the proposed RPACE. The extrinsic FPCA is a multivariate FPCA 
applied to the sparse manifold-valued data as if they are objects in the ambient Euclidean space. In the CFPCA approach 
one  first fits an FPCA to each of the $D$ components functions, and then applies a second PCA to the pooled component scores to obtain summarized scores and multidimensional eigenfunctions. For sample trajectories with small variation around the mean,  the  extrinsic FPCA methods (FPCA and CFPCA) can be regarded as linear approximations to  RPACE.
The Epanechnikov kernel was used for the smoothers, with bandwidth $h_\mu$ selected by GCV and $h_\Gamma=2h_\mu$. 

Using 200 Monte Carlo experiments, we report the average Root Mean Integrated Squared Errors (RMISE) for the fitted trajectories, defined as 

\[
\text{RMISE} = \sqrt{\frac{1}{200} \sum_{b=1}^{200} \frac{1}{D}\int_0^1 d_{\manifold}(\hat{X}_K(t), X(t))^2 dt},
\]
for $K=1, \dots, 6$ in \autoref{tab:fitted}. Since the fitted trajectories using FPCA and CFPCA lie in the ambient space but not on $\manifold$, we projected them back to the manifold by normalizing the norm of $\hat{X}_K(t)$ for $\manifold=S^2$ or the eigenvalues of the matrix representation of $\hat{X}_K(t)$ for $\manifold=\text{SO}(3)$.
RPACE was the overall best performer across the  various scenarios, especially for the  more parsimonious models.  Scenario 2 (sparse) is considerably more difficult than Scenario 1, and  smaller models with $K \le 4$ performed better. 
 RPACE also works well for the  smaller sample size $n=50$ in Scenario 3. 


\begin{table}[ht]
\centering
\caption{RIMSE for fitted trajectories. The standard errors were smaller than $3\times 10^{-3}$ for all cases: $X_E$, extrinsic Functional Principal Component Analysis \citep{chio:14}; $X_C$, componentwise Functional Principal Component Analysis \citep{happ:18}; $X_R$, proposed Riemannian Principal Component Analysis by Conditional Expectation (RPACE).}\vspace{.3cm}
\label{tab:fitted}
\begingroup\small\setlength{\tabcolsep}{0.3em}
\begin{tabular}{crrrrrr|rrrrrr|rrrrrr}
	    &                     \multicolumn{6}{c|}{Scenario 1 (baseline)}                      &                      \multicolumn{6}{c|}{Scenario 2 (sparse)}                       &                     \multicolumn{6}{c}{Scenario 3 (small $n$)}                     \\
	    & \multicolumn{3}{c}{$\manifold=S^2$} & \multicolumn{3}{c|}{$\manifold=\text{SO}(3)$} & \multicolumn{3}{c}{$\manifold=S^2$} & \multicolumn{3}{c|}{$\manifold=\text{SO}(3)$} & \multicolumn{3}{c}{$\manifold=S^2$} & \multicolumn{3}{c}{$\manifold=\text{SO}(3)$} \\
	$K$ & $X_E$ & $X_C$ & $X_R$ & \; $X_E$ & $X_C$ &                         $X_R$ & $X_E$ & $X_C$ &               $X_R$ & \; $X_E$ & $X_C$ &                         $X_R$ & $X_E$ & $X_C$ &               $X_R$ & \; $X_E$ & $X_C$ &                        $X_R$ \\
	\hline
	$1$ &  0.23 &  0.23 &           0.21 &  0.24 &  0.24 &                          0.22 &  0.26 &  0.27 &                0.24 &  0.26 &  0.27 &                          0.24 &  0.23 &  0.23 &                0.21 &  0.24 &  0.23 &                         0.21 \\
	$2$ &  0.12 &  0.12 &           0.09 &  0.12 &  0.12 &                          0.09 &  0.16 &  0.17 &                0.14 &  0.15 &  0.16 &                          0.12 &  0.12 &  0.12 &                0.09 &  0.12 &  0.12 &                         0.09 \\
	$3$ &  0.08 &  0.09 &           0.05 &  0.08 &  0.08 &                          0.04 &  0.13 &  0.15 &                0.11 &  0.11 &  0.13 &                          0.08 &  0.08 &  0.08 &                0.05 &  0.08 &  0.08 &                         0.04 \\
	$4$ &  0.05 &  0.06 &           0.04 &  0.05 &  0.05 &                          0.03 &  0.11 &  0.14 &                0.10 &  0.09 &  0.11 &                          0.07 &  0.05 &  0.06 &                0.04 &  0.05 &  0.05 &                         0.03 \\
	$5$ &  0.05 &  0.05 &           0.04 &  0.04 &  0.04 &                          0.02 &  0.10 &  0.13 &                0.10 &  0.08 &  0.10 &                          0.07 &  0.05 &  0.05 &                0.04 &  0.04 &  0.04 &                         0.02 \\
	$6$ &  0.04 &  0.05 &           0.04 &  0.03 &  0.04 &                          0.02 &  0.10 &  0.13 &                0.10 &  0.08 &  0.10 &                          0.07 &  0.04 &  0.05 &                0.04 &  0.03 &  0.04 &                         0.02

\end{tabular}
\endgroup
\end{table}

%
%
%
%
%
%
%
%
%
%
%

\section*{Appendix}
\subsubsection*{Proofs of Main Results}

\begin{proof}[Proof of Theorem \ref{thm:mu-asymptotic}]
According to Lemma \ref{lem:Qn-consistency}, where all auxiliary results are given in the next section, 
\begin{equation}
\sup_{p\in\manifold}\sup_{t\in\tdomain} |\tilde{Q}_{h_\mu}(p,t) - M(p, t)|  = O(h_\mu^2).
\end{equation}
With Lemma \ref{lem:L0-L2}, by similar arguments as in  Theorem 3 in \citet{mull:18:3},   
\begin{equation}
\sup_{t\in\tdomain}d_{\manifold}^{2}(\tilde{\mu}(t),\mu(t))=O(h_{\mu}^{4}) \label{eq:mean-bias}
\end{equation}
as $h_{\mu}\rightarrow0$ for $h_\mu = O(n^{-1/2})$. 

This result, combined with Lemma \ref{lem:variance-rate},
yields (\ref{eq:mean-rate}) for $m_{i}\equiv m$. The proof for the general
case follows  the same lines.
\end{proof}
\begin{proof}[Proof of Theorem \ref{thm:covariance-rate}]
We prove the theorem for $m_{i}\equiv m$, while the proof for the
general case is similar.  We will 
use $h$ to denote $h_{\Gamma}$ throughout the proof. Observe
\[
\hat{\Gamma}(s,t)=\frac{(S_{20}S_{02}-S_{11}^{2})R_{00}-(S_{10}S_{02}-S_{01}S_{11})R_{10}+(S_{10}S_{11}-S_{01}S_{20})R_{01}}{(S_{20}S_{02}-S_{11}^{2})S_{00}-(S_{10}S_{02}-S_{01}S_{11})S_{10}+(S_{10}S_{11}-S_{01}S_{20})S_{01}},
\]
where for $a,b=0,1,2$,
\begin{align*}
S_{ab} & =\sum_{i=1}^{n}v_{i}\sum_{1\leq j\neq l\neq m_{i}}K_{h}(T_{ij}-s)K_{h}(T_{il}-t)\left(\frac{T_{ij}-s}{h}\right)^{a}\left(\frac{T_{ij}-t}{h}\right)^{b},\\
R_{ab} & =\sum_{i=1}^{n}v_{i}\sum_{1\leq j\neq l\neq m_{i}}K_{h}(T_{ij}-s)K_{h}(T_{il}-t)\left(\frac{T_{ij}-s}{h}\right)^{a}\left(\frac{T_{ij}-t}{h}\right)^{b}\Gamma_{ijl}.
\end{align*}
Let $\delta_{ijl}=(\Log_{\mu(T_{ij})}Y_{ij})(\Log_{\mu(T_{il})}Y_{il})^{\transpose}$.
Then \newline 
$R_{00}=  R_{00}^{\prime}+(\Log_{\hat{\mu}(T_{ij})}Y_{ij}-\Log_{\mu(T_{ij})}Y_{ij})(\Log_{\hat{\mu}(T_{il})}Y_{il})^{\transpose}+\\
  (\Log_{\hat{\mu}(T_{ij})}Y_{ij})(\Log_{\hat{\mu}(T_{il})}Y_{il}-\Log_{\mu(T_{ij})}Y_{ij})^{\transpose}+\\
  (\Log_{\hat{\mu}(T_{ij})}Y_{ij}-\Log_{\mu(T_{ij})}Y_{ij})(\Log_{\hat{\mu}(T_{il})}Y_{il}-\Log_{\mu(T_{ij})}Y_{ij})^{\transpose},$
where
\begin{align*}
R_{ab}^{\prime} & =\sum_{i=1}^{n}v_{i}\sum_{1\leq j\neq l\neq m_{i}}K_{h}(T_{ij}-s)K_{h}(T_{il}-t)\left(\frac{T_{ij}-s}{h}\right)^{a}\left(\frac{T_{ij}-t}{h}\right)^{b}\delta_{ijl}.
\end{align*}
Given the smoothness of $\Log_{p}q$ with respect to $p$ and the
compactness of $\manifold$, 
\[
\|R_{00}-R_{00}^{\prime}\|_{F}^{2}\leq c\sup_{t\in\tdomain}d_{\manifold}^{2}(\hat{\mu}(t),\mu(t))=\Op(h_{\mu}^{4}+\frac{\log n}{n}+\frac{\log n}{nmh_{\mu}}),
\]
with similar results for $R_{ab}$ for $a,b=0,1,2$.
Setting  
\[
\tilde{\Gamma}(s,t)=\frac{(S_{20}S_{02}-S_{11}^{2})R_{00}^{\prime}-(S_{10}S_{02}-S_{01}S_{11})R_{10}^{\prime}+(S_{10}S_{11}-S_{01}S_{20})R_{01}^{\prime}}{(S_{20}S_{02}-S_{11}^{2})S_{00}-(S_{10}S_{02}-S_{01}S_{11})S_{10}+(S_{10}S_{11}-S_{01}S_{20})S_{01}},
\]
we have 
\[
\|\hat{\Gamma}(s,t)-\tilde{\Gamma}(s,t)\|_{F}^{2}=\Op(h_{\mu}^{4}+\frac{\log n}{n}+\frac{\log n}{nmh_{\mu}}),
\]
whence by the same argument as in  Theorem 5.2 in \citet{Zhang2016}, 
$\sup_{s,t\in\tdomain}\|\tilde{\Gamma}(s,t)-\Gamma(s,t)\|_{F}^{2}=\Op\{h^{4}+(\log n)(n^{-1}m^{-1}h^{-2}+n^{-1}h^{-1}+n^{-1})\}$,
and the result follows.
\end{proof}
\subsubsection*{Technical Lemmas}
\begin{lem}
\label{lem:Qn-consistency} Assume conditions \ref{cond:X0}, \ref{cond:X1}, \ref{cond:X2}, \ref{cond:M0},
\ref{cond:K0}, \ref{cond:L0}, \ref{cond:L1}  and
\ref{cond:H1} hold. Then for any $\epsilon>0$,
\begin{gather}
\inf_{t\in\tdomain}\inf_{\epsilon<d_{\manifold}(y,\mu(t))}\{M(y,t)-M(\mu(t),t)\}>0.\label{eq:inf-inf-1}\\
\sup_{p\in\manifold}\sup_{t\in\tdomain} |\tilde{Q}_{h_\mu}(p,t) - M(p,t)| =O(h^2_\mu)= o(1), \label{eq:Qn-consistency}\\
\sup_{t\in\tdomain}d_{\manifold}(\tilde{\mu}(t),\mu(t))=o(1).\label{eq:tmu-consistency}
\end{gather}

%
\end{lem}

\begin{proof}
Equation \eqref{eq:inf-inf-1} is implied by \ref{cond:L0}. For \eqref{eq:Qn-consistency}, first obtain the auxiliary result
\begin{equation}\label{eq:uu-rate}
u_{k}(t)=O(h_{\mu}^{k}),
\end{equation}
for $k=0,1,2$, where the $O(h_{\mu}^{k})$ term is uniform over  $t\in\tdomain$ and is bounded away from 0 for $k=0, 2$. This is due to change of variables and a Taylor expansion for  $f$,
\begin{align*}
\expect\{ K_{h_\mu}(T-t)\left(\frac{T-t}{h_\mu}\right)^k \} & = \int_{-t/h_\mu}^{(1-t)/h_\mu} s^k K(s) f(t + h_\mu s) ds \\
& = (f(t)+O(h))\int_{\max(-1,-t/h_\mu)}^{\min(1,(1-t)/h_\mu)} s^k K(s) ds \\
& = O(1),
\end{align*}
where the $O(1)$ term is uniform over $t$ and bounded away from 0 for $k=0,2$ by \ref{cond:L1}. 
Then
\begin{align}
\tilde{Q}_{h_\mu}(p,t) - M(p, t) & = \expect\{ \expect\{ d_\manifold^2(Y,p)|T\} \omega(T,t,h_\mu) \} - M(p, t)  \nonumber  \\
& = \expect\{ M(p,T) \omega(T,t,h_\mu) \} - M(p, t) \nonumber  \\
& = \expect\{ \frac{\partial}{\partial t}  M(p, t)(T-t)\omega(T,t,h_\mu) \} + 
\expect\{ \frac{\partial^2 }{\partial t^2} M(p,\vartheta) (T-t)^2 \omega(T,t,h) \}  \nonumber \\
& =  \frac{1}{u_2(t)u_0(t) - u_1(t)^2} [ u_2(t) \expect\{\frac{\partial^2 }{\partial t^2} M(p,\vartheta) K_{h_\mu}(T-t) (T-t)^2 \} - \nonumber\\
& \quad\quad u_0(t) \expect\{\frac{\partial^2 }{\partial t^2} M(p,\vartheta) K_{h_\mu}(T-t)(T-t)^3 \} ] \nonumber \\
& = O(h_\mu^2), \label{eq:tQn-rate}
\end{align}
where $\vartheta$ is between $T$ and $t$, the third equality is due to applying Taylor's theorem on $M(p, \cdot)$, the fourth to $\expect[(T-t)\omega(T,t,h_\mu)]=0$, and the last to \eqref{eq:uu-rate} and the continuity and boundedness in $\partial^2 M/\partial t^2(p,t)$, as  implied by \ref{cond:M0} and \ref{cond:X2}. Note that the rate \eqref{eq:tQn-rate} is uniform over $p\in \manifold$ and $t\in\tdomain$, so we obtain \eqref{eq:Qn-consistency}.

By M-estimation theory \citep[e.g. Corollary~3.2.3 in][]{vand:00}, \eqref{eq:Qn-consistency} and \eqref{eq:inf-inf-1} imply \eqref{eq:tmu-consistency}.
\end{proof}

%
%

\begin{lem}\label{lem:L0-L2}
Under conditions \ref{cond:M0}, \ref{cond:K0}, \ref{cond:L0}--\ref{cond:L2}, \ref{cond:X0}--\ref{cond:X2}, and $h_{\mu}\rightarrow0$, there exist constants $C>0$ and $\eta>0$ such that for all $t\in\tdomain$, 
\be
& M(y,t)-M(\mu(t),t) \geq  Cd^2_{\manifold}(y,\mu(t)), \label{eq:M-curvature}\\
& \underset{n}{\lim\inf}\{\tilde{Q}_{h}(y,t)-\tilde{Q}_{h}(\tilde{\mu}(t),t)-Cd^2_{\manifold}(y,\tilde{\mu}(t))\}  \geq  0, \label{eq:Qn-curvature}
\ee
if $d_{\manifold}(y,\mu(t))<\eta$ and $d_{\manifold}(y,\tilde{\mu}(t))<\eta$, respectively. 


\end{lem}

\begin{proof}
Recall $G_p(v,t) = M(\Exp_p v,t)$ as defined in \ref{cond:L2}, and define $v = \Log_{\mu(t)}(y)$. For each $t\in \tdomain$, apply a Taylor expansion  to obtain
\begin{align}
M(y,t)-M(\mu(t),t) & = G_{\mu(t)}(v,t) - G_{\mu(t)}(0,t) \nonumber \\
& = \innerprod{\frac{\partial^{2}}{\partial v^{2}}G_{\mu(t)}(v^*,t) v}v_{\mu(t)}\nonumber\\
& \geq  \lambda_{\min}(\frac{\partial^{2}}{\partial v^{2}}G_{\mu(t)}(v^*,t)) \innerprod vv_{\mu(t)} \nonumber\\
& =\lambda_{\min}(\frac{\partial^{2}}{\partial v^{2}}G_{\mu(t)}(v^*,t)) d^2_\manifold(y,\mu(t)) \nonumber \label{eq:pf-inf-inf-1}
\end{align}
where $v^*$ is between 0 and $v$. There exists $\eta > 0$ such that for $\innerprod{v}{v}^{1/2} < \eta$, 
\begin{equation}
\lambda_{\min}(\frac{\partial^{2}}{\partial v^{2}}G_{\mu(t)}(v^*,t)) \ge C, \label{eq:C}
\end{equation}
by \ref{cond:L2} and the smoothness of $G_{\mu(t)}$, where $C=\lambda_{\min}(\partial^2G_{\mu(t)}/\partial v^2 (0,t))/2$, and the inequality holds uniformly over $t$. This then implies \eqref{eq:M-curvature}.

For \eqref{eq:Qn-curvature}, applying iterated expectations and Taylor's theorem, we  obtain
\begin{equation}
\frac{\partial^{2}}{\partial y^{2}}\tilde{Q}_{h}(y,t) - \frac{\partial^{2}}{\partial y^{2}}M(y,t) = O(h_\mu^2), \label{eq:ddQn-rate}
\end{equation}
where the $O(h_\mu^2)$ term is uniform over $y\in \manifold$ and $t\in \tdomain$,
similar to the proof of the uniform consistency of $\tilde{Q}_{h_\mu}$ to $M(p,t)$ in Lemma~\ref{lem:Qn-consistency}. 
Define $H_{p}(v,t)=\tilde{Q}_{h}(\Exp_{p}v,t)$ and $\tilde{v} = \Log_p(y)$. With \ref{cond:L2} this implies
\begin{equation}
\inf_{t\in\tdomain}\underset{n\rightarrow\infty}{\lim\inf}\,\lambda_{\min}\left(\frac{\partial^{2}}{\partial v^{2}}H_{\tilde{\mu}(t)}(v,t)\mid_{v=0}\right) = \inf_{t\in\tdomain}\underset{n\rightarrow\infty}{\lim\inf}\,\lambda_{\min}\left(\frac{\partial^{2}}{\partial y^{2}}\tilde{Q}_{h}(y,t)\mid_{y=\tilde{\mu}(t)}\right)  >0, \label{eq:inf-inf-2-1}
\end{equation}  
where the inequality is due to \eqref{eq:ddQn-rate}. Then
\begin{align*}
\tilde{Q}_{h}(y,t)-\tilde{Q}_{h}(\tilde{\mu}(t),t) & = H_{\tilde{\mu}(t)}(\tilde{v}, t) - H_{\tilde{\mu}(t)}(0, t)\\
& = \innerprod{\frac{\partial^2}{\partial v^2} H_{\tilde{\mu}(t)}(\tilde{v}^*, t) \tilde{v}}{\tilde{v}}_{\tilde{\mu}(t)}\\
&\ge \lambda_{\min}(\frac{\partial^2}{\partial v^2} H_{\tilde{\mu}(t)}(\tilde{v}^*, t)) \innerprod{\tilde{v}}{ \tilde{v}}_{\tilde{\mu}(t)}\\
& = \lambda_{\min}(\frac{\partial^2}{\partial v^2} H_{\tilde{\mu}(t)}(\tilde{v}^*, t)) d_\manifold^2(y, \tilde{\mu}(t))\\
& = \lambda_{\min}(\frac{\partial^2}{\partial v^2} H_{\tilde{\mu}(t)}(\tilde{v}^*, t)) d_\manifold^2(y, \tilde{\mu}(t))
\end{align*}
By \eqref{eq:C} and \eqref{eq:ddQn-rate}, the last term is not smaller than $C d_\manifold^2(y,\tilde{\mu}(t))$ for large enough $n$. Therefore by taking liminf and infimum over $t$ we obtain \eqref{eq:Qn-curvature}. 
\end{proof}

\begin{lem}
\label{lem:entropy}Suppose $B_{\delta}(p)$ is an open ball centered
at $p\in\manifold$ with radius $\delta>0$, 
and denote the covering number of $B_{\delta}(p)$ with $\epsilon$-balls by $N(\epsilon,B_{\delta}(p),d_{\manifold})$. Then  condition \ref{cond:M0} implies 
\[
\int_{0}^{1}\sup_{t\in\tdomain}\sqrt{1+\log N(\delta\epsilon,B_{\delta}(\mu(t)),d_{\manifold})}\diffop\epsilon=O(1)\quad\text{as}\quad\delta\rightarrow0.
\]
\end{lem}
\begin{proof}
This is a consequence of Proposition 3 of \citet{mull:18:3}. 
\end{proof}

\begin{lem}
\label{lem:variance-rate}Suppose $m_{i}=m$ and \ref{cond:M0}, \ref{cond:K0},
\ref{cond:L0}, \ref{cond:L1}, \ref{cond:X0}, \ref{cond:X1} and \ref{cond:X2} hold. If $h_{\mu}\rightarrow 0$
and $nmh_{\mu}\rightarrow\infty$, then
\[
\sup_{t\in\tdomain}d_{\manifold}^{2}(\hat{\mu}(t),\tilde{\mu}(t))=\Op\left(\frac{\log n}{n}+\frac{\log n}{nmh_{\mu}}\right).
\]
\end{lem}

\begin{proof}

We first establish 
\begin{equation}
\sup_{t\in\tdomain}|\hat{u}_{k}(t)-u_{k}(t)|=\Op\left(\sqrt{\log n}\sqrt{h_{\mu}^{2k-1}n^{-1}m^{-1}+h_{\mu}^{2k}n^{-1}}\right),  \label{eq:u-rate}
\end{equation}
with proof analogous to that of Lemma 5 of \citet{Zhang2016}.
Following a similar argument as in  Lemma 2 in \citet{mull:18:3}, 
by \eqref{eq:u-rate}, Lemma~\ref{lem:L0-L2} and  \ref{lem:entropy},
one obtains
\begin{equation}
\sup_{t\in\tdomain}d_{\manifold}^{2}(\hat{\mu}(t),\tilde{\mu}(t))=\op(1).\label{eq:uniform-consistency}
\end{equation}

To derive the rate of convergence, we set $\hat{\omega}_{ij}(t)=K_{h}(T_{ij}-t)\{\hat{u}_{0}-\hat{u}_{1}(T_{ij}-t)\}/\hat{\sigma}_{0}^{2}$
and $\tilde{\omega}_{ij}(t)=K_{h}(T_{ij}-t)\{u_{0}-u_{1}(T_{ij}-t)\}/\sigma_{0}^{2}$.
With $S_{n}(y,t)=\hat{Q}_{n}(y,t)-\tilde{Q}_{h}(y,t)$ and $D_{ij}(y,t)=d_{\manifold}^{2}(Y_{ij},y)-d_{\manifold}^{2}(Y_{ij},\tilde{\mu}(t))$,
\begin{align*}
|S_{n}(y,t)-S_{n}(\tilde{\mu}(t),t)| & \leq\left|\frac{1}{nm}\sum_{i=1}^{n}\sum_{j=1}^{m}\{\hat{\omega}_{ij}(t)-\tilde{\omega}_{ij}(t)\}D_{ij}(y,t)\right|\\
& \quad +\left|\frac{1}{nm}\sum_{i=1}^{n}\sum_{j=1}^{m}[\tilde{\omega}_{ij}(t)D_{ij}(y,t)-\expect\{\tilde{\omega}_{ij}(t)D_{ij}(y,t)\}]\right|\\
& \equiv A_{1}(y,t)+A_{2}(y,t).
\end{align*}
For any $\delta>0$, using the boundedness of $d_{\manifold}$ and
\eqref{eq:u-rate}, one can deduce that \newline $\sup_{t\in\tdomain}\sup_{d_{\manifold}(y,\tilde{\mu}(t))<\delta}A_{1}(y,t)=\Op(\delta\sqrt{\log n}/\sqrt{nmh})$,
with a universal constant for all $\delta>0$. Thus, for 
\[
B_{R}=\left\{ \sup_{t\in\tdomain}\sup_{d_{\manifold}(y,\tilde{\mu}(t))<\delta}\left|\frac{1}{nm}\sum_{i=1}^{n}\sum_{j=1}^{m}\{\hat{\omega}_{ij}(t)-\tilde{\omega}_{ij}(t)\}D_{ij}(y,t)\right|\leq R\delta\sqrt{\log n}/\sqrt{nmh}\right\} 
\]
for some $R>0$, it holds  that $\prob(B_{R}^{c})\rightarrow0$. For the
second term, we employ a similar argument as in Lemma 5 of \citet{Zhang2016}
to show that $\expect\sup_{t\in\tdomain}\sup_{d_{\manifold}(y,\tilde{\mu}(t))<\delta}A_{2}(y,t)=O(\delta\sqrt{\log n}\sqrt{1/n+1/nmh})$. 
Thus,
\[
\expect\left\{ I_{B_{R}}\sup_{t\in\tdomain}\sup_{d_{\manifold}(y,\tilde{\mu}(t))<\delta}|S_{n}(y,t)-S_{n}(\tilde{\mu}(t),t)|\right\} \leq a\delta\sqrt{\log n}\sqrt{1/n+1/nmh},
\]
where $a$ is a constant depending on $R$. To finish, set $r_{n}=(\sqrt{\log n}\sqrt{1/n+1/nmh})^{-1}$
and define $S_{k,n}(t)=\{y:2^{k-1}\leq r_{n}d(y,\tilde{\mu}(t))\leq2^{k}\}$.
Let $\eta$ be as in Lemma \ref{lem:L0-L2}, and $\tilde{\eta}=\eta/2$.
Then for any positive integer $W$,
\begin{align*}
\prob\left(\sup_{t\in\tdomain}r_{n}d_{\manifold}(\tilde{\mu}(t),\hat{\mu}(t))>2^{W}\right)\leq & \prob(B_{R}^{c})+\prob(2\sup_{t\in\tdomain}d_{\manifold}(\tilde{\mu}(t),\hat{\mu}(t))>\eta_{2})\\
 & \hspace{-2cm} +\sum_{\stackrel{k\geq W}{2^{k}\leq r_{n}\tilde{\eta}}}\prob\left[\left\{ \sup_{t\in\tdomain}\sup_{y\in S_{k,n}(t)}|S_{n}(y,t)-S_{n}(\tilde{\mu}(t),t)|\geq c\frac{2^{2(k-1)}}{r_{n}^{2}}\right\} \cap B_{R}\right],
\end{align*}
where $c>0$ is some constant, and the second term goes to zero for
any $\eta>0$ according to (\ref{eq:uniform-consistency}). Since
$d_{\manifold}(y,\tilde{\mu}(t))\leq2^{k}/r_{n}$ on $S_{k,n}(t)$,
this implies that the sum on the right-hand side of the above inequality
is bounded by 
\[
4ac\sum_{\stackrel{k\geq W}{2^{k}\leq r_{n}\tilde{\eta}}}\frac{2^{-k}}{r_{n}^{-1}}\sqrt{\log n}\sqrt{1/n+1/nmh}\leq\sum_{k\geq W}2^{k} \rightarrow 0
\]
as $W\rightarrow 0$. Therefore, 
\[
\sup_{t\in\tdomain}d_{\manifold}^{2}(\hat{\mu}(t),\tilde{\mu}(t))=\Op\left(\frac{\log n}{n}+\frac{\log n}{nmh_{\mu}}\right).
\]
\end{proof}

\single

\references


\begin{thebibliography}{40}
\expandafter\ifx\csname natexlab\endcsname\relax\def\natexlab#1{#1}\fi

\bibitem[{Aitchison(1986)}]{aitc:86}
\textsc{Aitchison, J.} (1986).
\newblock \textit{The {S}tatistical {A}nalysis of {C}ompositional {D}ata}.
\newblock London: Chapman \& Hall.

\bibitem[{Anirudh et~al.(2017)Anirudh, Turaga, Su \& Srivastava}]{anir:17}
\textsc{Anirudh, R.}, \textsc{Turaga, P.}, \textsc{Su, J.} \&
  \textsc{Srivastava, A.} (2017).
\newblock Elastic functional coding of {R}iemannian trajectories.
\newblock \textit{IEEE Transactions on Pattern Analysis and Machine
  Intelligence} \textbf{39}, 922--936.

\bibitem[{Bhattacharya \& Bhattacharya(2012)}]{bhat:12}
\textsc{Bhattacharya, A.} \& \textsc{Bhattacharya, R.} (2012).
\newblock \textit{Nonparametric inference on manifolds: with applications to
  shape spaces}, vol.~2.
\newblock Cambridge University Press.

\bibitem[{Bhattacharya \& Patrangenaru(2003)}]{Bhattacharya2003}
\textsc{Bhattacharya, R.} \& \textsc{Patrangenaru, V.} (2003).
\newblock Large sample theory of intrinsic and extrinsic sample means on
  manifolds. {I}.
\newblock \textit{The Annals of Statistics} \textbf{31}, 1--29.

\bibitem[{Bosq(2000)}]{Bosq2000}
\textsc{Bosq, D.} (2000).
\newblock \textit{Linear Proceses in Function Spaces}.
\newblock Lecture Notes in Statistics. Springer.

\bibitem[{Chen \& M\"{u}ller(2012)}]{Chen2012}
\textsc{Chen, D.} \& \textsc{M\"{u}ller, H.} (2012).
\newblock Nonlinear manifold representations for functional data.
\newblock \textit{The Annals of Statistics} \textbf{40}, 1--29.

\bibitem[{Chen \& Lei(2015)}]{chen:15:1}
\textsc{Chen, K.} \& \textsc{Lei, J.} (2015).
\newblock Localized functional principal component analysis.
\newblock \textit{Journal of the American Statistical Association}
  \textbf{110}, 1266--1275.

\bibitem[{Chiou et~al.(2014)Chiou, Chen \& Yang}]{chio:14}
\textsc{Chiou, J.-M.}, \textsc{Chen, Y.-T.} \& \textsc{Yang, Y.-F.} (2014).
\newblock Multivariate functional principal component analysis: {A}
  normalization approach.
\newblock \textit{Statistica Sinica} \textbf{24}, 1571--1596.

\bibitem[{Dai \& M{\"u}ller(2018)}]{dai:17:1}
\textsc{Dai, X.} \& \textsc{M{\"u}ller, H.-G.} (2018).
\newblock Principal component analysis for functional data on {R}iemannian
  manifolds and spheres.
\newblock \textit{Annals of Statistics} \textbf{46}, 3334--3361.

\bibitem[{Egozcue et~al.(2003)Egozcue, Pawlowsky-Glahn, Mateu-Figueras \&
  Barcel{\'o}-Vidal}]{egoz:03}
\textsc{Egozcue, J.~J.}, \textsc{Pawlowsky-Glahn, V.}, \textsc{Mateu-Figueras,
  G.} \& \textsc{Barcel{\'o}-Vidal, C.} (2003).
\newblock Isometric logratio transformations for compositional data analysis.
\newblock \textit{Mathematical Geology} \textbf{35}, 279--300.

\bibitem[{Fan \& Gijbels(1996)}]{fan:96}
\textsc{Fan, J.} \& \textsc{Gijbels, I.} (1996).
\newblock \textit{Local Polynomial Modelling and its Applications}.
\newblock London: Chapman \& Hall.

\bibitem[{Fang et~al.(1990)Fang, Kotz \& Ng}]{fang:90}
\textsc{Fang, K.}, \textsc{Kotz, S.} \& \textsc{Ng, K.} (1990).
\newblock \textit{Symmetric multivariate and related distributions}.
\newblock Monographs on statistics and applied probability. London: Chapman and
  Hall.

\bibitem[{Hall \& Horowitz(2007)}]{Hall2007c}
\textsc{Hall, P.} \& \textsc{Horowitz, J.~L.} (2007).
\newblock Methodology and convergence rates for functional linear regression.
\newblock \textit{The Annals of Statistics} \textbf{35}, 70--91.

\bibitem[{Hall \& Hosseini-Nasab(2006)}]{Hall2006}
\textsc{Hall, P.} \& \textsc{Hosseini-Nasab, M.} (2006).
\newblock On properties of functional principal components analysis.
\newblock \textit{Journal of the Royal Statistical Society: Series B
  (Statistical Methodology)} \textbf{68}, 109--126.

\bibitem[{Hall et~al.(2006)Hall, M\"{u}ller \& Wang}]{Hall2006a}
\textsc{Hall, P.}, \textsc{M\"{u}ller, H.-G.} \& \textsc{Wang, J.-L.} (2006).
\newblock Properties of principal component methods for functional and
  longitudinal data analysis.
\newblock \textit{The Annals of Statistics} \textbf{34}, 1493--1517.

\bibitem[{Happ \& Greven(2018)}]{happ:18}
\textsc{Happ, C.} \& \textsc{Greven, S.} (2018).
\newblock Multivariate functional principal component analysis for data
  observed on different (dimensional) domains.
\newblock \textit{Journal of the American Statistical Association}
  \textbf{113}, 649--659.

\bibitem[{Horvath \& Kokoszka(2012)}]{horv:12}
\textsc{Horvath, L.} \& \textsc{Kokoszka, P.} (2012).
\newblock \textit{Inference for Functional Data with Applications}.
\newblock New York: Springer.

\bibitem[{Hsing \& Eubank(2015)}]{Hsing2015}
\textsc{Hsing, T.} \& \textsc{Eubank, R.} (2015).
\newblock \textit{Theoretical Foundations of Functional Data Analysis, with an
  Introduction to Linear Operators}.
\newblock Wiley.

\bibitem[{Kendall et~al.(2009)Kendall, Barden, Carne \& Le}]{kend:09}
\textsc{Kendall, D.}, \textsc{Barden, D.}, \textsc{Carne, T.} \& \textsc{Le,
  H.} (2009).
\newblock \textit{Shape and Shape Theory}.
\newblock Hoboken: Wiley.

\bibitem[{Kleffe(1973)}]{Kleffe1973}
\textsc{Kleffe, J.} (1973).
\newblock Principal components of random variables with values in a separable
  {H}ilbert space.
\newblock \textit{Statistics: A Journal of Theoretical and Applied Statistics}
  \textbf{4}, 391--406.

\bibitem[{Kneip et~al.(2016)Kneip, Poss \& Sarda}]{knei:16}
\textsc{Kneip, A.}, \textsc{Poss, D.} \& \textsc{Sarda, P.} (2016).
\newblock Functional linear regression with points of impact.
\newblock \textit{Annals of Statistics} \textbf{44}, 1--30.

\bibitem[{Kong et~al.(2016)Kong, Xue, Yao \& Zhang}]{kong:16}
\textsc{Kong, D.}, \textsc{Xue, K.}, \textsc{Yao, F.} \& \textsc{Zhang, H.~H.}
  (2016).
\newblock Partially functional linear regression in high dimensions.
\newblock \textit{Biometrika} \textbf{103}, 147--159.

\bibitem[{Kraus(2015)}]{krau:15}
\textsc{Kraus, D.} (2015).
\newblock Components and completion of partially observed functional data.
\newblock \textit{Journal of the Royal Statistical Society: Series B
  (Statistical Methodology)} \textbf{77}, 777--801.

\bibitem[{Krueger \& Mueller(2011)}]{krue:11}
\textsc{Krueger, A.~B.} \& \textsc{Mueller, A.} (2011).
\newblock Job search, emotional well-being, and job finding in a period of mass
  unemployment: Evidence from high-frequency longitudinal data.
\newblock \textit{Brookings Papers on Economic Activity} \textbf{2011}, 1--57.

\bibitem[{Lang(1995)}]{Lang1995}
\textsc{Lang, S.} (1995).
\newblock \textit{Differential and Riemannian Manifolds}.
\newblock New York: Springer.

\bibitem[{Lee(1997)}]{Lee1997}
\textsc{Lee, J.~M.} (1997).
\newblock \textit{Riemannian Manifolds: An Introduction to Curvature}.
\newblock New York: Springer-Verlag.

\bibitem[{Li(2015)}]{li:15}
\textsc{Li, H.} (2015).
\newblock Microbiome, metagenomics, and high-dimensional compositional data
  analysis.
\newblock \textit{Annual Review of Statistics and Its Application} \textbf{2},
  73--94.

\bibitem[{Li \& Hsing(2010)}]{Li2010}
\textsc{Li, Y.} \& \textsc{Hsing, T.} (2010).
\newblock Uniform convergence rates for nonparametric regression and principal
  component analysis in functional/longitudinal data.
\newblock \textit{The Annals of Statistics} \textbf{38}, 3321--3351.

\bibitem[{Patrangenaru et~al.(2018)Patrangenaru, Bubenik, Paige \&
  Osborne}]{patra:18}
\textsc{Patrangenaru, V.}, \textsc{Bubenik, P.}, \textsc{Paige, R.~L.} \&
  \textsc{Osborne, D.} (2018).
\newblock Topological data analysis for object data.
\newblock \textit{arXiv preprint arXiv:1804.10255} .

\bibitem[{Petersen \& M\"{u}ller(2018)}]{mull:18:3}
\textsc{Petersen, A.} \& \textsc{M\"{u}ller, H.-G.} (2018).
\newblock Fr\'echet regression for random objects with {E}uclidean predictors.
\newblock \textit{Annals of Statistics, to appear (arXiv preprint
  arXiv:1608.03012)} .

\bibitem[{Talsk{\'a} et~al.(2018)Talsk{\'a}, Menafoglio, Machalov{\'a}, Hron \&
  Fi{\v{s}}erov{\'a}}]{tals:18}
\textsc{Talsk{\'a}, R.}, \textsc{Menafoglio, A.}, \textsc{Machalov{\'a}, J.},
  \textsc{Hron, K.} \& \textsc{Fi{\v{s}}erov{\'a}, E.} (2018).
\newblock Compositional regression with functional response.
\newblock \textit{Computational Statistics \& Data Analysis} \textbf{123},
  66--85.

\bibitem[{Telschow et~al.(2016)Telschow, Huckemann \&
  Pierrynowski}]{Telschow2016}
\textsc{Telschow, F.~J.}, \textsc{Huckemann, S.~F.} \& \textsc{Pierrynowski,
  M.~R.} (2016).
\newblock Functional inference on rotational curves and identification of human
  gait at the knee joint.
\newblock \textit{arXiv:1611.03665} .

\bibitem[{Thompson(1942)}]{thom:42}
\textsc{Thompson, D.} (1942).
\newblock \textit{On Growth and Form}.
\newblock Dover Publications.

\bibitem[{{van der}~Vaart \& Wellner(1996)}]{vand:00}
\textsc{{van der}~Vaart, A.} \& \textsc{Wellner, J.} (1996).
\newblock \textit{Weak Convergence and Empirical Processes: With Applications
  to Statistics}.
\newblock New York: Springer.

\bibitem[{Verbeke et~al.(2014)Verbeke, Fieuws, Molenberghs \&
  Davidian}]{verb:14}
\textsc{Verbeke, G.}, \textsc{Fieuws, S.}, \textsc{Molenberghs, G.} \&
  \textsc{Davidian, M.} (2014).
\newblock The analysis of multivariate longitudinal data: A review.
\newblock \textit{Statistical Methods in Medical Research} \textbf{23}, 42--59.

\bibitem[{Wang et~al.(2016)Wang, Chiou \& M\"uller}]{mull:16:3}
\textsc{Wang, J.-L.}, \textsc{Chiou, J.-M.} \& \textsc{M\"uller, H.-G.} (2016).
\newblock Functional data analysis.
\newblock \textit{Annual Review of Statistics and its Application} \textbf{3},
  257--295.

\bibitem[{Yao et~al.(2005)Yao, M\"uller \& Wang}]{Yao2005a}
\textsc{Yao, F.}, \textsc{M\"uller, H.-G.} \& \textsc{Wang, J.-L.} (2005).
\newblock Functional data analysis for sparse longitudinal data.
\newblock \textit{Journal of the American Statistical Association}
  \textbf{100}, 577--590.

\bibitem[{Yuan et~al.(2012)Yuan, Zhu, Lin \& Marron}]{Yuan2012}
\textsc{Yuan, Y.}, \textsc{Zhu, H.}, \textsc{Lin, W.} \& \textsc{Marron, J.~S.}
  (2012).
\newblock Local polynomial regression for symmetric positive definite matrices.
\newblock \textit{Journal of Royal Statistical Society: Series B (Statistical
  Methodology)} \textbf{74}, 697--719.

\bibitem[{Zhang \& Chen(2007)}]{Zhang2007}
\textsc{Zhang, J.-T.} \& \textsc{Chen, J.} (2007).
\newblock Statistical inferences for functional data.
\newblock \textit{The Annals of Statistics} \textbf{35}, 1052--1079.

\bibitem[{Zhang \& Wang(2016)}]{Zhang2016}
\textsc{Zhang, X.} \& \textsc{Wang, J.~L.} (2016).
\newblock From sparse to dense functional data and beyond.
\newblock \textit{The Annals of Statistics} \textbf{44}, 2281--2321.

\end{thebibliography}
\end{document}